\documentclass[11pt]{article}

\usepackage[top=0.8in, bottom=0.9in, left=1in, right=1in]{geometry}
\usepackage{setspace, times}
\usepackage{amsmath, amssymb, amsthm, graphicx}
\usepackage{amsmath, dsfont, amscd, latexsym, epsfig, color, times, ifthen, epstopdf}
\usepackage{amssymb, amsfonts}
\usepackage{amsthm}

\usepackage[justification=centering]{caption}

\spacing{1.125}
\usepackage{url}
\usepackage{times}

 \def\11{\mathbf{1}}

\def\setof#1{\left\{{\let\st\colon #1 }\right\}}

 \def\calT{\mathcal{T}}

 \def\BOO={=^{\hspace{0.06cm}\epsilon}_B}
  
  \def\00{\mathbf{0}}
\def\aa{\mathbf{a}} \def\bb{\mathbf{b}}  
\def\uu{\mathbf{u}} \def\vv{\mathbf{v}}

\def\pp{\mathbf{p}} \def\qq{\mathbf{q}}

 \def\calR{\mathcal{R}} \def\calU{\mathcal{U}}

  \newtheorem{theo}{Theorem}

 \def\opt{\text{\sf OPT}}

\newtheorem{lemma}{Lemma}[section]
\newtheorem{property}{Property}[section]
\newtheorem{corollary}{Corollary}[section]

\newtheorem{definition}{Definition}[section]

\newcommand {\N} {{\rm I\kern-2.5pt N}}
\newcommand {\R} {{\rm I\kern-2.5pt R}}

\def\nil{\text{nil}}

\renewcommand{\paragraph}[1]{\noindent {\bf #1}}
 \def\pp{\mathbf{p}}

\newcommand{\eps}{\epsilon}

  \def\aa{\mathbf{a}} \def\bb{\mathbf{b}}
\def\qq{\mathbf{q}}

\newcommand{\ignore}[1]{}

 \def\vv{\mathbf{v}} \def\uu{\mathbf{u}}

\def\opt{\textsf{OPT}}  \def\Uti{\textsf{Uti}}
\def\topU{\textsf{topU}}

\begin{document}

\title{The Complexity of Optimal Multidimensional Pricing\vspace{1cm}}
\author{Xi Chen\thanks{Columbia University. Email: 
  \texttt{\{\hspace{0.03cm}xichen, xiaoruisun\hspace{0.03cm}\}\hspace{0.03cm}@cs.columbia.edu}. Research supported by NSF grant CCF-1149257 and a Sloan fellowship.} \and
Ilias Diakonikolas\thanks{University of Edinburgh. Email: \texttt{ilias.d@ed.ac.uk}. Research supported in part by a startup from the University of Edinburgh and a SICSA PECE grant.} \and
Dimitris Paparas\thanks{Columbia University. Email: \texttt{\{\hspace{0.03cm}paparas, 
  mihalis\hspace{0.03cm}\}\hspace{0.03cm}@cs.columbia.edu}. Research supported by NSF grant CCF-1017955.} \and
Xiaorui Sun$^{\ast}$ \and
Mihalis Yannakakis$^{\ddagger}$
}

\date{}

\setcounter{page}{0}

\maketitle 

\thispagestyle{empty}

\begin{abstract}
We resolve the complexity of revenue-optimal deterministic auctions in 
  the unit-demand single-buyer Bayesian setting, i.e., the optimal item 
  pricing problem, when the buyer's values for the items are inde\-pendent. 
We show that the problem of computing a revenue-optimal pricing can be solved
  in polynomial time for distributions of support size $2$, 
  and {its decision version} is NP-complete for distributions of support size $3$. 
We also show that the problem remains NP-complete for the case of identical distributions.
\end{abstract}
\newpage


\section{Introduction}
 
Consider the following natural pricing scenario: We have a set of $n$ items for sale and a single {\em unit-demand} buyer, i.e., 
a consumer interested in obtaining at most one of the items. 
The goal of the seller is then to set prices for the items in order to maximize her revenue
by exploiting stochastic information about~the buyer's preferences. More specifically, the seller is given
access to a distribution $\mathcal{F}$ from which the buyer's valuations $\vv = (v_1, \ldots, v_n)$ for
  the items are drawn, 
i.e.,  $\vv \sim \mathcal{F}$, and wants to assign a price $p_i$ to each~item~in 
  order to maximize her expected revenue. 
We assume, as is commonly the case, that the buyer here is~{\em quasi-linear}, i.e., her utility 
  for item $i\in [n]$ is $v_i - p_i$, and 
  she will select an item with the maximum nonnegative utility or nothing if no such item exists. 
This is known as the {\em Bayesian Unit-demand Item-Pricing Pro\-blem (BUPP)}~\cite{ChK07}, 
  and has received considerable attention
in the CS literature during the past few years \cite{GHK+05, ChK07, Briest08, ChHMS10, CD11, DDT12a}.

Throughout this paper we focus on the well-studied case~\cite{ChK07, ChHMS10, CD11}
that $\mathcal{F} = \times_{i=1}^n \mathcal{F}_i$ is a {\em product distribution}, 
i.e., the valuations of the buyer for the items are mutually independent random variables. We assume that the $n$ (marginal) distributions
$\mathcal{F}_i$ are discrete and are known to the seller (i.e., the values of the support and the corresponding probabilities are rational numbers
given explicitly in the input). This seemingly simple computational problem appears to exhibit a very rich structure. Prior to our work, even the (very special) case
that the distributions $\mathcal{F}_i$ have support $2$ was 
not well understood: First note that the search space is apparently exponential, 
since the support size of $\mathcal{F}$ is $2^n$. What makes things trickier is that the optimal prices are not necessarily in the support of $\mathcal{F}$ 
(see~\cite{CD11} for a simple example with two items with distributions of support $2$). So, a priori, it was not even clear whether the optimal prices can be described
with polynomially many bits in the size of the input description.

Revenue-optimal pricing is well-studied by economists (see, e.g.,~\cite{W96} for a survey and~\cite{MMW89} for a simple additive case with two items). 
The pricing problem studied in this work fits in the general frame\-work of {\em optimal multi-dimensional mechanism design}, 
a central question in mathematical economics (see~\cite{ManelliV07} and references therein). 
Finding the optimal deterministic mechanism in our setting is equivalent to finding the optimal item-pricing. 
A randomized mechanism, on the other hand, would allow the seller to price lotteries over items~\cite{BCKW10, ChMS10}, albeit this may be less natural in this context.

Optimal mechanism design is well-understood in single-parameter settings for which Myerson~\cite{M81} gives a closed-form characterization for the optimal mechanism. 
Chawla, Hartline and Kleinberg~\cite{ChK07} show that techniques
from Myerson's work can be used to obtain an analogous closed-form characterization 
  (and also an efficient algorithm) for pricing in our setting, albeit with a constant factor
loss in the revenue. In particular, they obtain a factor $3$ approximation to the optimal
  expected revenue (subsequently improved to $2$ in~\cite{ChHMS10}).
Cai and Daskalakis~\cite{CD11} obtain a polynomial-time approximation scheme for distributions with monotone hazard-rate
(and a quasi-polynomial time approximation scheme for the broader class of regular distributions). That is, prior to this work, closed-form characterizations
(and efficient algorithms) were known for {\em approximately} optimal pricing.
The question of whether such a characterization exists for the {\em optimal} pricing has remained open and was posed as an open problem in these works~\cite{ChK07, CD11}.

\medskip 

\noindent {\bf Our Results.} In this paper, we take a principled complexity-theoretic look at the BUPP with independent (discrete) distributions. 
We start by showing (Theorem~\ref{theo:membership}) that the general decision 
  problem is in NP~(and 
  as a corollary, the optimal prices can be described with polynomially many bits). We note that the membership proof is non-trivial because the
  optimal prices may not be in the support. Our proof proceeds by partitioning the space of price-vectors into a set of (exponentially many) cells (defined by the value distributions $\mathcal{F}_i$), so that the optimal revenue within each cell can be found efficiently by a shortest path computation. One consequence of the analysis is that the optimal pricing problem has the integrality property: 
  if the values in the supports are
  integer then the optimal prices are also integer (though they may not belong to the support).

We then proceed to show (Theorem~\ref{theo:supp-2}) that the case in which each marginal distribution has support
   at most  $2$ can be solved in polynomial time. 
Indeed, by exploiting the underlying structure of the problem, we show that 
  it suffices to consider $O(n^2)$ price-vectors to compute the optimal revenue in this case.

{\em Our main result is that the problem is NP-hard, even for distributions of support $3$ 
  \emph{(Theorem~\ref{theo:np-hard-3})} or distributions that are identical but have
  large support \emph{(Theorem~\ref{theo:np-hard-iid})}.} This answers an open problem~first posed in~\cite{ChK07} and also asked in~\cite{CD11, DDT12a}.
The main difficulty in the reductions stems from the fact that, for a general instance of the pricing problem, the expected revenue is a highly complex nonlinear function of the prices.
The challenge is to construct an instance such that the revenue can be well-approximated by a simple function 
  and is also general enough to encode an NP-hard problem.


\medskip

\noindent {\bf Previous Work.} We have already mentioned the main algorithmic works for the independent distributions case 
with approximately-optimal revenue guarantees~\cite{ChK07, ChHMS10, CD11}.
On the lower bound side, Guruswami et al.~\cite{GHK+05} and subsequently Briest~\cite{Briest08} studied the complexity of the problem when the buyer's values for the items 
are {\em correlated}, respectively obtaining APX-hardness and $\Omega(n^{\eps})$  inapproximability, for some constant $\eps>0$.
More recently, Daskalakis, Deckelbaum and Tzamos~\cite{DDT12a} showed that the pricing problem with independent distributions is SQRT-SUM-hard when either the support values or the probabilities are irrational. We note that 
  their reduction relies
  on the fact that, for certain carefully constructed instances,
  it is SQRT-SUM-hard to compare the revenue of two price-vectors.
This has no bearing on the complexity of the problem under the standard discrete model we consider,
for which the exact revenue of a price-vector can be computed efficiently. 

\medskip

\noindent {\bf Related Work.} 
The optimal mechanism design problem (i.e., the problem of finding a revenue-maximizing mechanism in a Bayesian setting) has received considerable attention in the CS community during the past few years. The vast majority of the work so far is algorithmic~\cite{ChK07, ChHMS10, BGGM10, Alaei11, DFK11, HN12, CDW12a, CDW12b}, providing approximation or exact algorithms for various versions of the problem. Regarding lower bounds, Papadimitriou and Pierrakos~\cite{PP11}
show that computing the optimal deterministic single-item auction is APX-hard, even for the case of $3$ bidders. We remark that, if randomization is allowed, this problem 
can be solved exactly in polynomial time via linear programming~\cite{DFK11}.
In very recent work,  Daskalakis, Deckelbaum and Tzamos~\cite{DDT12b} show $\#P$-hardness for computing the optimal randomized mechanism for the case of additive buyers.
We remark that their result does not have any implication for the unit-demand case
due to the very different structures of the two problems.
\medskip

The rest of the paper is organized as follows.
In Section 2 we first define formally the problem, state our main results,
and prove some preliminary basic properties. 
In Section 3 we show that the decision problem is in NP.
In Section 4 we give a polynomial-time algorithm for distributions
with support size 2. Section 5 shows NP-hardness for the case of support size 3,
and Section 6 for the case of identical distributions.
We conclude in Section 7.

\def\itemp{\textsc{Item-Pricing }}
\def\iidp{\textsc{Iid-Item-Pricing }}
\def\knap{\textsc{Integer Knapsack with repetitions }}
\def\Revv{\mathcal{R}}
\def\prob{\Pr}
\def\Util{\mathcal{U}}
\def\g{\gamma}
\def\G{\Gamma}

\section{Preliminaries}

\subsection{Problem Definition and Main Results}

In our setting, there are one buyer and one seller with $n$ items, indexed by 
  $[n]=\{1,2,\ldots,n\}$. 
The buyer~is interested in buying at most one item {(unit demand)}, and
  her valuation of the items are 
  drawn from $n$ inde\-pendent discrete distributions, one for each item.
In particular, we use $V_i=\{v_{i,1},\ldots,v_{i,|V_i|}\}$, 
  {$i \in [n]$,} to denote the support of 
the value distribution of item $i$, where $0\le v_{i,1}<\cdots<v_{i,|V_i|}$.
We also use $q_{i,j}>0$, {$j \in [|V_i|]$,} to denote the
probability of item $i$ having value $v_{i,j}$, with $\sum_{j} q_{i,j}=1$.
Let $V= {\times_{i=1}^n  V_i}$. We use $\prob[\vv]$ to denote
  the probability of the valuation vector being $\vv=(v_1,\ldots,v_n)\in V$,
  i.e., the product of $q_{i,j}$'s over $i,j$ such that $i\in [n]$
  and $v_i=v_{i,j}$.

In the problem, all the $n$ distributions, i.e., $V_i$ and $q_{i,j}$, 
  are given to the seller explicitly. 
The seller then assigns a price $p_i\ge 0$ to each item.
Once the price vector $\pp=(p_1,\ldots,p_n)\in \mathbb{R}_+^n$ is fixed,
the buyer draws her values $\vv=(v_1,\ldots,v_n)$ from
the $n$ distributions independently, i.e., $\vv\in V$ with probability $\prob[\vv]$.
{We assume that the buyer is quasi-linear, i.e., her utility for item $i$ equals $v_i - p_i$.}
Let 
\begin{eqnarray*}
&\Util(\vv,\pp)=\max_{i\in [n]} \left(v_i-p_i\right).&
\end{eqnarray*}
If $\Util(\vv,\pp)\ge 0$, the buyer selects an item $i \in [n]$ that  
maximizes her utility $v_i-p_i$, and the revenue of the seller is $p_i$.
If $\calU(\vv,\pp)<0$, the buyer does not select any item, 
  and the revenue of the seller is $0$. 

Knowing the value distributions as well as the behavior of the buyer described above, 
the seller's objective is to {compute} a price vector $\pp\in \mathbb{R}_+^n$ that maximizes the expected revenue 
\[
\Revv(\pp)=\sum_{i\in[n]} p_i\cdot \Pr\big[\hspace{0.06cm}
\text{buyer selects item $i$}\hspace{0.05cm}\big].
\]
We use $\itemp$ to denote the following decision problem:
The input consists of $n$ discrete distributions, with $v_{i,j}$
and $q_{i,j}$ all being rational and encoded in binary, and a rational number $t \ge 0$. 
The problem asks whether the supremum of the expected revenue $\Revv(\pp)$ over all 
  price vectors $\pp\in \mathbb{R}_{+}^n$ is at least $t$, 
  where we use $\mathbb{R}_+$ to denote the set of nonnegative real numbers. 

{We note that} the aforementioned decision problem is not well-defined without a tie-breaking rule,
i.e., {a rule that specifies} which item the buyer {selects} when there are multiple 
items with {maximum nonnegative utility.}
Throughout the paper, we will use the following \emph{maximum 
  price}\footnote{It may also be called the \emph{maximum value} tie-breaking rule,
  since an item with the maximum price among a set of items with\\ the 
  same utility must also have the maximum value.} 
tie-breaking rule {(which is convenient for our arguments)}: 
when there are multiple items {with maximum nonnegative utility}, the buyer
{selects the item with the smallest index among items with the highest price.
(We note that the critical part is that an item with the highest price is selected.
Selecting the item with the smallest index among them is arbitrary --- and does not affect  the revenue;
however we need to make such a choice so that it makes sense to talk about ``the'' item selected by the buyer in the proofs.)}
We show in Section \ref{tie} that our choice of the tie-breaking rule does not affect the supremum of the expected revenue
(hence,  the complexity of the problem).

We are now ready to state our main results.
First,  we show in Section \ref{npmember-sec}
  that \itemp\ is in NP.

\begin{theo}\label{theo:membership}
\itemp\ is in NP.
\end{theo}

Second, we present in Section \ref{support2-sec} a polynomial-time
  algorithm for \itemp\ when all the distributions have support size
  at most $2$. 

\begin{theo} \label{theo:supp-2}
\itemp\ is in {P} when every distribution has support size
  at most $2$.
\end{theo}

{As our main result, we resolve the computational complexity of the problem.} 
We show that it is NP-hard even when all distributions have
support size at most $3$ (Section \ref{nphard-sec}), or when they are identical (Section~\ref{iid-sec}).
\begin{theo} \label{theo:np-hard-3}
\itemp\ is NP-hard even when every distribution has
  support size at most $3$.
\end{theo} 

\begin{theo} \label{theo:np-hard-iid}
\itemp\ is NP-hard even when the distributions are identical.
\end{theo}



\def\topp{\mathcal{T}}
\def\best{\mathcal{I}}

\subsection{Tie-Breaking Rules}\label{tie}

In this section, we show that the supremum of the expected revenue
over $\pp\in \mathbb{R}_{+}^n$ is invariant to tie-breaking rules.
{Formally}, a tie-breaking rule is a {mapping from the set of pairs} $(\vv,\pp)$
  with $\Util(\vv,\pp)\ge 0$ to an item $k$ such that
  $v_k-p_k=\Util(\vv,\pp)$.

{We will need some notation. Let $B$ be the maximum price tie-breaking rule
   described earlier.
We will denote by $\Revv(\pp)$ the expected revenue of $\pp$ under $B$, and 
by $\Revv(\vv,\pp)$ the seller's revenue under $B$ when the valuation vector is $\vv\in V$.
Given a price vector $\pp$ and a valuation vector $\vv\in V$, we also denote by
$\topp(\vv,\pp)$ the set of items with  maximum nonnegative utility
  (so $\topp(\vv,\pp)=\emptyset$\hspace{0.03cm} iff \hspace{0.03cm}$\Util(\vv,\pp)<0$).
}

We show the following:

\begin{lemma}\label{tie-breaking-ha}
The supremum of the expected revenue over $\pp\in \mathbb{R}_{+}^n$
  is invariant to tie-breaking rules.
\end{lemma}
\begin{proof}
Let $v_{i,j}$ and $q_{i,j}$ denote the numbers that specify the distributions.
Let $B'$ be a tie-breaking rule.
We will use $\Revv'(\pp)$ to denote the expected revenue
of $\pp$ under $B'$ and use $\Revv'(\vv,\pp)$ to denote 
the seller's revenue under $B'$ when the valuation vector is $\vv\in V$.


It is clear that for any $\pp\in \mathbb{R}_{+}^n$
  and $\vv\in V$, we have $\Revv(\vv,\pp)\ge \Revv'(\vv,\pp)$ 
since $B$ picks an item with the highest price among those that maximize the utility.
Hence, it follows that $\sup_\pp \Revv(\pp)\ge \sup_\pp \Revv'(\pp)$.

On the other hand, given any price vector $\pp\in \mathbb{R}_{+}^n$, we 
  consider 
$$\pp_{\epsilon}=\big(\max(0,p_1-r_1\epsilon),\ldots,\max(0,p_n-r_n\epsilon)\big)\in \mathbb{R}_{+}^n,$$
  where $\epsilon>0$ and $r_i$ is the rank of $p_i$ sorted in increasing order
  (when there are ties, the item with the smaller index is ranked higher).
We claim that 
\begin{equation}
\label{eq:1}\lim_{\epsilon\rightarrow 0+} \Revv'(\pp_{\epsilon})=\Revv(\pp).
\end{equation}
It then follows {from (\ref{eq:1})} that $\sup_\pp \Revv'(\pp)\ge \sup_\pp \Revv(\pp)$,
 {which gives the proof of the lemma}.

To prove (\ref{eq:1}), we show that the following holds for
   any valuation vector $\vv\in V$:
\begin{equation}\label{eq:2}
\lim_{\epsilon\rightarrow 0+} \Revv'(\vv,\pp_\epsilon)=\Revv(\vv,\pp).
\end{equation}
{Observe that} (\ref{eq:1}) follows from (\ref{eq:2}) since 
$$
\calR(\pp)=\sum_{\vv\in V} \calR(\vv,\pp)\cdot \Pr[\vv]\ \ \ \ \text{and}\ \ \ \
\calR'(\pp_\epsilon)=\sum_{\vv\in V} \calR'(\vv,\pp_\epsilon)\cdot \Pr[\vv].
$$
  
To prove (\ref{eq:2}), we consider two cases. 
If $\Util(\vv,\pp)<0$, then we have $\Util(\vv,\pp_\epsilon)<0$ 
  when $\epsilon$ is sufficiently small, and thus, $\Revv(\vv,\pp)=\Revv'(\vv,\pp_\epsilon)=0$.
When $\Util(\vv,\pp)\ge 0$,
we make the following three observations.
First, the utility of an item $i\in [n]$ under $\pp_\epsilon$ is at least as high as that under $\pp$.
Second, if $v_i-p_i>v_j-p_j$ for some items $i,j\in [n]$, then under $\pp_\epsilon$ the utility
of item $i$ remains strictly higher than that of item $j$, for $\epsilon$ sufficiently small. 
Third, if 
  $v_i-p_i=v_j-p_j$ and $p_i>p_j$ (in particular, $p_i>0$) 
  for some $i,j\in [n]$, then under $\pp_{\epsilon}$
  the utility of item $i$ is strictly higher than that 
  of item $j$ when $\epsilon\ll p_i$, as $r_i>r_j$.
It follows from these observations that when $\epsilon$ is sufficiently
  small, $B'$ must pick, given $\vv$ and $\pp_\epsilon$,
  an item $k\in [n]$ such that $p_k=\Revv(\vv,\pp)$.
(\ref{eq:2}) then follows from the definition of $\pp_\epsilon$.
\end{proof}

We will {henceforth} always adopt the maximum price tie-breaking rule,  
  and use $\calR(\vv,\pp)$ to denote the revenue of the seller with respect to this rule.
One of the advantages of this rule is that the supremum of 
  the expected revenue $\calR(\pp)$ is always achievable, so 
  it makes sense to talk about whether $\pp$ is optimal or not.
In the following example, we point out that this does not hold for general tie-breaking rules.\vspace{0.07cm}

\medskip

\noindent {\bf Example:} Suppose item $1$ has value $10$ with probability $1$, 
item $2$ has value $8$ with probability $1/2$ and value $12$ with probability $1/2$, 
and in case of tie the buyer prefers item $1$. 
The supremum in this example is $11$: set $p_1=10$ for item $1$ 
and $p_2=12-\epsilon$ for item $2$. 
The buyer will buy item $1$ with probability $1/2$ (if her value for item $2$ is $8$) 
and item $2$ with probability $1/2$ (if her value for item $2$ is $12$). 
However, {an expected revenue} of $11$ is not achievable: if we give price $12$ to item $2$, 
then the  buyer will always buy item $1$ and the revenue is $10$. 
Note that the expected revenue for this tie-breaking rule is not a continuous function of the prices.\vspace{0.07cm}

\medskip

Before proving that the supremum is indeed always achievable under the maximum
  price rule, we start by showing that 
without loss of generality, we may focus the search for an optimal price vector in the set
\begin{eqnarray*}
&P= \times_{i=1}^n [a_i,b_i],
\ \ \ \ \text{where\hspace{0.05cm}
  $a_i=\min_{j} v_{i,j}$ \hspace{0.05cm}{and}\hspace{0.05cm} $b_i=\max_j v_{i,j}$}&
\end{eqnarray*}
denote the minimum and
  maximum values in the support $V_i$, respectively.

\begin{lemma}\label{searchspace}
For any price vector  $\pp \in  \mathbb{R}^n_{+}$, 
  there exists a $\pp' \in  P$ such that $\Revv (\pp') \ge \Revv(\pp).$ 
\end{lemma}
\begin{proof}
First, it is straightforward that no price $p_i$ should be above $b_i$; 
  if such a price exists, we can simply replace it by $b_i$ and this 
  {will not decrease} the expected revenue.

The non-trivial part is to argue that it is no loss of generality to assume that no price $p_i$ is below $a_i$. 
Let $\pp \in \times_{i=1}^n [0, b_i].$ Suppose that there exists $i \in [n]$ such that $p_i < a_i$, i.e., the set 
$L(\pp) = \{\hspace{0.03cm} i \in [n]: p_i < a_i 
  \hspace{0.03cm}\}$ is nonempty; otherwise, there is nothing to prove.

Fix an $i \in L(\pp)$ arbitrarily and let $S_i = \{ j \in [n] : p_j < a_i \}.$
We consider the price vector $\widetilde{\pp}$ defined by $\widetilde{p}_j = \min \{b_j, a_i \}$ for $j \in S_i$ and $\widetilde{p}_j = p_j$ otherwise.
As $i \in S_i$, it follows that $S_i \neq \emptyset$ and therefore $\widetilde{\pp} \neq \pp$
  (in particular, $\widetilde{p}_i=a_i$ now). It is also clear that $\widetilde{\pp} \in \times_{i=1}^n [0, b_i].$
It suffices to show that $\Revv (\widetilde{\pp}) \ge \Revv(\pp)$.

Indeed, note that $|L(\widetilde{\pp})| < |L(\pp)|$ so this process will terminate in at most $n$ stages.
After the last stage we will obtain a vector $\pp' \in P$ whose expected revenue is lower bounded by all the previous ones.

To prove that $\Revv (\widetilde{\pp}) \ge \Revv(\pp)$, we {proceed} as follows. 
Given any valuation vector $\vv \in V$, we compare the revenue $\Revv(\vv, \pp)$ to 
$\Revv(\vv, \widetilde{\pp})$ and consider the following two cases:
\begin{flushleft}
\begin{itemize}
\item {\bf Case 1:} 
  On input $(\vv,\pp)$, the item selected by the buyer is not from $S_i$.
We claim that the same item is selected on input $(\vv,\widetilde{\pp})$.
   Indeed, we did not decrease prices of items in $S_i$, hence their 
   utilities did not go up, while the utilities of the remaining 
   items did not change. 
Therefore, 
  the revenue does not change in this case, i.e., 
  $\Revv (\vv, \widetilde{\pp}) = \Revv(\vv, \pp)$.\vspace{-0.1cm}

\item {\bf Case 2:} On input $(\vv,\pp)$, the item selected is from $S_i$. 
Then by the definition of $S_i$, the revenue $\Revv(\vv,\pp)$ we get is certainly less than $a_i$. 
On input $(\vv, \widetilde{\pp})$, we know that $\Util(\vv,\widetilde{\pp})\ge 0$ (since item $i$ 
  must have nonnegative utility, i.e., $v_i - \widetilde{p}_i = v_i - a_i \ge 0$) and thus, 
  $\topp(\vv,\widetilde{\pp})\ne \emptyset$.
 We claim that $\Revv(\vv, \widetilde{\pp}) \ge a_i>\Revv(\vv,\pp)$. 
To see this, we consider two sub-cases.   
If $\Util(\vv,\widetilde{\pp}) = 0$, then we must have $i \in \topp(\vv,\widetilde{\pp})$ and the claim follows
  from our choice of the maximum price tie-breaking rule. 
If $\Util(\vv,\widetilde{\pp}) > 0$, then 
  every $j\in \topp(\vv,\widetilde{\pp})$ must satisfy $\widetilde{p}_j\ge a_i$; 
  otherwise, by definition of $\widetilde{\pp}$ we have
  $\widetilde{p}_j=b_j$ and $v_j-\widetilde{p}_j\le 0$, a contradiction.
From $\widetilde{p}_j\ge a_i$ and $j\in \topp(\vv,\widetilde{\pp})$,
  we have $\Revv(\vv,\widetilde{\pp})\ge a_i$. 
\end{itemize}
\end{flushleft}
The lemma follows by combining the two cases.
\end{proof}

Now we show that the supremum 
  can always be achieved {under the maximum price rule $B$}.

\begin{lemma}\label{supremum}
There exists a price vector $\pp^*\in P$ such that ${\Revv}(\pp^*)=
  \sup_{\pp} {\Revv}(\pp)$.
\end{lemma}
\begin{proof}
{By the compactness of $P$,}
  it suffices to show that if a sequence of vectors 
  $\{\pp_i\}$ approaches $\pp$, then 
$$\Revv(\pp)\ge \lim_{i\rightarrow \infty} \Revv(\pp_i).$$
To this end, it suffices to show that, for any valuation vector $\vv\in V$, 
\begin{equation}\label{eq:33}
\Revv(\vv,\pp)\ge \lim_{i\rightarrow \infty} \Revv(\vv,\pp_i).
\end{equation}
Given any valuation $\vv\in V$, it is easy to check that $\calT(\vv,\pp_i)\subseteq \calT(\vv,\pp)$
  when $i$ is sufficiently large. (Again consider two cases: $\Util(\vv,\pp)<0$ and $\Util(\vv,\pp)\ge 0$.)
(\ref{eq:33}) then follows, since $\calR(\vv,\pp)$ is
  the highest price of all items in $\calT(\vv,\pp)$ under the maximum price
  tie-breaking rule.
\end{proof}

\section{Membership in NP}\label{npmember-sec}

In this section we prove Theorem \ref{theo:membership}, i.e.,
  \itemp\ is in NP.


\begin{proof}[Proof of Theorem \ref{theo:membership}]
We start with some notation.
Given a price vector $\pp\in \mathbb{R}_+^n$ and a valuation
  $\vv\in V$,~let $\best(\vv,\pp)\in [n]\cup \{\nil\}$
  denote the item picked by the buyer under the maximum price
  tie-breaking rule, with $\best(\vv,\pp)=\nil$\hspace{0.02cm} iff 
  \hspace{0.02cm}$\Util(\vv,\pp)<0$. 
We will partition $P=\times_{i=1}^n [a_i,b_i]$ into equivalence classes so 
  that two price vectors $\pp,\pp'$ from the same class yield the 
  same outcome for all valuations: $\best(\vv,\pp)=\best(\vv,\pp')$ for all $\vv$. 

Consider the partition of $P$ induced by the following set of hyperplanes. 
For each item $i\in [n]$ and each value $s_i\in V_i$, we have a hyperplane $p_i = s_i$. 
For each pair of items $i,j\in [n]$ and pair of values $s_i\in V_i$ and $t_j\in V_j$, we have a hyperplane 
$s_i - p_i = t_j - p_j$, i.e., $p_i -p_j = s_i - t_j$. 
These hyperplanes partition our search space $P$ into polyhedral cells, 
where the points in each cell lie on the same side of each hyperplane 
(either on the hyperplane or in one of the two open-halfspaces). 

We claim that, for every valuation $\vv\in V$, all the 
  vectors in each cell yield the same outcome. 
Consider any cell $C$. It is defined by a set of equations and inequalities.
Given any price vector $\pp \in C$ and any value $s_i \in V_i$,
let $V(\pp,s_i)$ be the set of valuation vectors $\vv\in V$ such that $v_i=s_i$
and the buyer ends up buying item $i$ on $(\vv,\pp)$. 
We claim that $V(\pp,s_i)$ does not
depend on $\pp$, i.e., it is the same set $V(s_i)=V(\pp,s_i)$ over all $\pp \in  C$.
To this end, first, if the points of $C$ satisfy $p_i > s_i$ then $V(\pp,s_i)= \emptyset$.
So suppose that $C$ satisfies $\pp \leq s_i$.
Consider any valuation vector $\vv\in V$ with $v_i=s_i$. The valuation $\vv$ is in $V(\pp,s_i)$ 
iff for all $j \neq i$, we have $s_i - p_i \geq v_j - p_j$,
and in case of equality we have $s_i \geq v_j$ 
  (iff $p_i\ge p_j$ due to the equality), and in case of further equality
$s_i = v_j$ we have $i < j$.
Because all points of the cell $C$ lie on the same side of each
hyperplane $s_i - p_i = v_j - p_j$, it follows that $V(\pp,s_i)$ does not depend on $\pp$.
As a result,
  for any cell $C$ and any $\vv\in V$,
  all the points $\pp\in C$ yield the same outcome $\best(\vv,\pp)$.

Next, we show that
  it is easy to compute the supremum of the expected revenue $\Revv(\pp)$ 
  over $\pp\in C$, for each cell $C$.
To this end, let $W_i= \cup_{s_i\in V_i} V(s_i)\subseteq V$ 
  denote the set of valuations
for which the buyer picks item $i$ if the prices lie in the cell $C$, 
and let $\gamma_i$ be the probability of $W_i$: $\gamma_i=\sum_{\vv\in W_i}\Pr[\vv].$
It turns out that $\gamma_i$ can be computed efficiently, since the 
  probability of $V(s_i)$ can be computed efficiently as shown below 
  (and $W_i$ is the disjoint union of $V(s_i)$, $s_i\in V_i$).

Given $s_i\in V_i$, to compute the probability of $V(s_i)$, 
  we note that $V(s_i)$ is actually the Cartesian product of subsets of $V_j$, $j\in [n]$.
For each $j\ne i$, we can determine efficiently 
  the subset of values $L_j\subseteq V_j$ such that the buyer prefers 
  item $i$ to $j$ if $i$ has value $s_i$ and $j$ has value from $L_j$.
As a result, we have $$V(s_i)=L_1\times \cdots \times L_{i-1}\times \{s_i\}\times L_{i+1}\times \cdots\times L_n,$$
and thus, we multiply the probabilities of these subsets $L_j$, for all $j$, 
  and the probability of $s_i$. 
Summing up the probabilities of $V(s_i)$ over $s_i\in V_i$ gives us $\gamma_i$,
  the probability of $W_i$.  

Finally, the supremum of the expected revenue
  $\Revv(\pp)$ over all $\pp\in C$
  is the maximum of $\sum_{i\in [n]} \gamma_i\cdot p_i$ over all $\pp$
  in the closure of $C$.
Let ${C'}$ denote the closure of $C$;
this is the polyhedron obtained by changing all the strict inequalities of $C$
into weak inequalities.
The supremum of $\sum_i \gamma_i\cdot p_i$ over all points $\pp \in C$
can be computed in polynomial time by solving the 
linear program that maximizes $\sum \gamma_i\cdot p_i$ subject to
$\pp \in {C'}$.
In fact, as we will show below after the proof of Theorem \ref{theo:membership}, 
  that this LP has a special form:
The question of whether a set of equations and inequalities with
respect to a set of hyperplanes of the form $p_i=s_i$
and $p_i -p_j = s_i - t_j$ is consistent, i.e., defines a nonempty cell,
can be formulated as a negative weight cycle problem,
and the optimal solution for a nonempty cell can be computed 
by solving a single-source shortest path problem. 
It follows that the specification of a cell $C$ in the partition
  is an appropriate \textit{yes} certificate for the decision problem
  \itemp, and the theorem is proved.
\end{proof}

Next we describe in more detail how to determine whether a set of equations
  and inequalities defines
  a nonempty cell, and how to compute the optimal solution over a nonempty cell.
The description of a~(candidate) cell $C$ consists of equations and inequalities
specifying (1) for each item $i$, the relation of $p_i$ to
every value $s_i \in V_i$, and (2) for each pair of items $i,j$
and each pair of values $s_i \in V_i$ and $t_j \in V_j$, the
relation of $p_i-p_j$ to $s_i - t_j$.
Construct a weighted directed graph $G=(N,E)$ over $n+1$ nodes
$N=\{0,1,\ldots, n \}$ where nodes $1,\ldots,n$ correspond
to the $n$ items.
For each inequality of the form $p_i< s_i$
or $p_i\le s_i$, include an edge $(0,i)$ with weight $s_i$, 
and call the edge strict or weak
accordingly as the inequality is strict or weak.
In fact, there is a tightest such inequality
(i.e., with the smallest value $s_i$) since the cell is in $P$, and it suffices to include
the edge for this inequality only. 
Similarly, for each inequality of the form 
$p_i> s_i$ or $p_i\ge s_i$ (or only for the tightest such inequality, 
i.e. the one with the largest value $s_i$) include an edge $(i,0)$
with weight $-s_i$.
For each inequality of the form $p_i-p_j< s_i - t_j$
or $p_i-p_j\le  s_i - t_j$ (or only for the tightest such inequality)
include a (strict or weak) edge $(j,i)$ with weight $s_i - t_j$.
Similarly, for every inequality of the form $p_i-p_j> s_i - t_j$
or $p_i-p_j\ge s_i - t_j$ (or only for the tightest such inequality)
include a (strict or weak) edge $(i,j)$ with weight $t_j-s_i$.

We prove the following connections between $G=(N,E)$
  and the cell $C$:

\begin{lemma}\label{lem:technical}
1. A set of equations and inequalities defines a nonempty cell if and only if
the corresponding\\ graph $G$ does not contain a negative weight cycle or
a zero weight cycle with a strict edge.\newline
\indent 2. The supremum of the expected revenue for a nonempty cell is achieved by the
price vector $\pp$ that\\ consists of the distances from node 0 to the other
nodes of the graph $G$.
\end{lemma}
\begin{proof}
1. Considering node 0 as having an associated variable $p_0$ with fixed value 0,
the given set of equations (i.e., pairs of weak inequalities) and (strict)
inequalities can be viewed as a set of difference constraints on the variables 
  $(p_0,p_1,\ldots,p_n)$, and
it is well known that the feasibility of such a set of constraints
can be formulated as a negative weight cycle problem.
If there is a cycle with negative weight $w$, then adding all the
inequalities corresponding to the edges of the cycle yields the
constraint $0 \leq w$ (which is false); if there is a cycle 
  with zero weight but also
a strict edge, then summing the inequalities yields $0 < 0$.

Conversely, suppose that $G$ does not contain a negative weight cycle
or a zero weight cycle with a strict edge. For each strict edge $e$,
replace its weight $w(e)$ by $w'(e)=w(e) - \epsilon$ for a sufficiently small $\epsilon>0$
(we can treat $\epsilon$ symbolically), and let $G(\epsilon)$
be the resulting weighted graph.
Note that $G(\epsilon)$ does not contain any negative weight cycle,
hence all shortest paths are well-defined in $G(\epsilon)$.
Compute the shortest (minimum weight)
paths from  node $0$ to all the other nodes in $G(\epsilon)$, and let $\pp(\epsilon)$ be
the vector of distances from $0$. 
For each edge $(i,j)$ the distances $p_i(\epsilon)$ and $p_j(\epsilon)$
(where $p_0(\epsilon)=0$) must satisfy $p_j(\epsilon)\leq  p_i(\epsilon) + w'(i,j),$ 
hence all the (weak and strict) inequalities are satisfied.

To determine if a set of equations and inequalities defines
a nonempty cell, we can form the graph $G(\epsilon)$ and test
for the existence of a negative weight cycle using for example
the Bellman-Ford algorithm.

\medskip
2. Suppose that cell $C$ specified by the constraints is nonempty. 
Then we claim that the vector $\pp = \pp(0)$
of distances from node 0 to the other nodes in the graph $G$
is greater than or equal
to any vector $\pp' \in C$ in all coordinates.
We can show this by induction on the depth of a node in the
shortest path tree $T$ of $G$ rooted at node $0$.
Letting $p'_0 = p_0=0$, the basis is trivial.
For the induction step, consider a node $j$ with parent $i$ in $T$.
By the inductive hypothesis $p'_i \leq p_i$. The edge
$(i,j)$ implies that $p'_j -p'_i \leq w(i,j)$ {or} {$<w(i,j)$},
and the presence of the edge $(i,j)$ in the shortest path tree implies that
$p_j=p_i+w(i,j)$.
Therefore, $p'_j \leq p_j$.

The supremum of the expected revenue $\Revv(\pp')$ over the cell $C$ is
given by the optimal value of the linear program
  that maximizes $\sum_{i\in [n]} \gamma_i\cdot p'_i$ subject to
$\pp' \in {C'}$, where ${C'}$ is the closure of the cell
$C$.
Observe that all the coefficients $\gamma_i$ of the objective function
are nonnegative, and clearly $\pp$ is in the closure ${C'}$.
Therefore $\pp$ achieves the supremum of the expected revenue over $C$.
\end{proof}

The NP characterization of \itemp\ and the corresponding structural characterization
  of the optimal price vector $\pp=\pp(0)$ of each cell have several easy 
  and useful consequences.

First, we get an alternative proof of Lemma \ref{supremum} regarding
  the maximum tie-breaking rule:



\begin{proof}[Second Proof of Lemma \ref{supremum}]
Suppose that the supremum of the expected revenue is achieved in cell $C$.
Let $G$ be the corresponding graph, and let $\pp$ be the price vector of the
distances from node 0 to the other nodes.
If $\pp \in C$ then the conclusion is immediate,
so assume $\pp \notin C$. From the proof of the above lemma we have
that $\pp \geq \pp'$ coordinate-wise for all $\pp' \in C$.

We claim that for any valuation $\vv\in V$, the revenue $\Revv(\vv,\pp)$
is at least as large as the revenue $\Revv(\vv,\pp')$ 
under any $\pp' \in C$.
Suppose that the buyer selects item $i$ under $\vv$ for prices $\pp'$.
Then $p'_i \leq v_i$ and thus also $p_i \leq v_i$ (since $\pp$ is
in the closure of $C$) and thus $i$ is also eligible for
selection under $\pp$. If the buyer selects $i$ under $\pp$
then we know that $p_i \geq p'_i$ and the conclusion follows.
Suppose that the buyer selects another item $j$ under $\pp$
and that $p'_i > p_j$ and hence $p_i > p_j$.
Then we must have $v_j - p_j > v_i -p_i$ 
  due to the tie-breaking rule.
The facts that $\pp$ is in the closure of $C$ and
$v_j - p_j > v_i -p_i$ imply that $v_j - p'_j > v_i -p'_i$
for all $\pp' \in C$, and therefore the buyer should
have picked $j$ instead of $i$ under prices $\pp'$, a contradiction.

We conclude that for any $\vv\in V$, $\Revv(\vv,\pp)\ge \Revv(\vv,\pp')$
  for any $\pp' \in C$, and the lemma follows.
\end{proof}

Another consequence 
  suggested by the structural characterization of Lemma \ref{lem:technical} is that  
  the maximum of expected revenue can always be achieved by a price vector $\pp$ in which
  all prices $p_i$ are sums of a value and differences between 
  pairs of values of items. This implies for example the following useful corollary.

\begin{corollary}\label{integerprice}\begin{flushleft}
If all the values in $V_i$, $i\in [n]$, are integers, then 
  there must exist an optimal price vector $\pp\in P$ with integer coordinates.
\end{flushleft}\end{corollary}


\def\Uti{\mathcal{U}}
\def\topU{\mathcal{T}}

\section{A polynomial-time algorithm for support size 2}\label{support2-sec}

In this section, we present a polynomial-time algorithm for the case 
  that each distribution has support size at most $2$.
In Section~\ref{ssec:special-case}, we give a polynomial-time algorithm under a certain ``non-degeneracy''
assumption on the values. In Section~\ref{ssec:general-case} we generalize this algorithm to handle the general case.

\subsection{An Interesting Special case.} \label{ssec:special-case}

In this subsection, we assume that every item has support size $2$,
  where $V_i=\{a_i,b_i\}$ satisfies
  $b_i> a_i> 0$, for all $i \in [n]$. 
Let $q_i:0<q_i<1$ denote the probability of 
  the value of item $i$ being $b_i$.
For convenience, we also let $t_i=b_i-a_i>0$.
In addition, we assume in this subsection that the value-vectors
  $\aa=(a_1,\ldots,a_n)$ and $\bb=(b_1,\ldots,b_n)$ satisfy the following
  ``non-degeneracy'' assumption:
\begin{quote}\centering
\textbf{Non-degeneracy assumption}: $b_1<b_2<\cdots<b_n$, $a_i\ne a_j$ {and} 
  $t_i\ne t_j$ {for all $i,j\in[n]$.}
\end{quote}
As we show next in Section \ref{ssec:general-case}, this special case encapsulates the essential difficulty of the problem.

Let $\opt$ denote the set of optimal price vectors in $P=\times_{i=1}^{n} [a_i,b_i]$ that 
  maximize the expected revenue $\Revv(\pp)$.
Next we prove a sequence of lemmas to show that, given $\aa$
  and $\bb$ that satisfy all the conditions above one can 
  compute efficiently a set $A\subseteq P$
  of price vectors such that $|A|=O(n^2)$ and $\opt\subseteq A$.  
Hence, by computing $\Revv(\pp)$ for all $\pp\in A$, we get 
  both the maximum of expected revenue and an optimal price vector.
  
We start with the following lemma:

\begin{lemma}\label{firstlem}
If $\pp\in P$ satisfies $p_i>a_i$ for all $i\in [n]$,
  then either $\pp=\bb$ or we have $\pp\notin \emph{\opt}$.
\end{lemma}
\begin{proof}
Assume for contradiction that $\pp\in P$ satisfies $p_i>a_i$, 
  for all $i\in [n]$ but $\pp\ne \bb$.
It then follows from the maximum price tie-breaking rule that  
  $\Revv(\vv,\bb)\ge \Revv(\vv,\pp)$
for all $\vv\in V$. 
Moreover, there is at least one 
  $\vv^*\in V$ such that
  $\Revv(\vv^*,\bb)>\Revv(\vv^*,\pp)$: If $p_i<b_i$, then consider $\vv^*$
  with $v^*_i=b_i$ and $v^*_j=a_j$ for all other $j$.
It follows that $\Revv(\bb)>\Revv(\pp)$ as we assumed that $0<q_i<1$ for
  all $i\in [n]$ and thus, $\pp\notin \opt$.
\end{proof} 

Next we show that 
  there can be at most one $i$ such that $p_i=a_i$; otherwise $\pp\notin \opt$.
We emphasize that all the conditions on $V_i$ are assumed in the lemmas below,
  the non-degeneracy assumption in particular.

\begin{lemma}\label{secondlem}
If $\pp\in P$ has more than one $i\in [n]$
  such that $p_i=a_i$, then we have $\pp\notin \emph{\opt}$.
\end{lemma}
\begin{proof}
Assume for contradiction that $\pp\in P$ has more than one $i$
  such that $p_i=a_i$.
We prove the lemma by explicitly constructing a new price vector $\pp'\in P$ from $\pp$ such that
  $\Revv(\vv,\pp')\ge \Revv(\vv,\pp)$ for all $\vv\in V$ and 
  $\Revv(\vv^*,\pp')>\Revv(\vv^*,\pp)$ for at least one $\vv^*\in V$.
This implies that $\Revv(\pp')>\Revv(\pp)$ and thus, 
  $\pp$ is not optimal.
We will be using this simple strategy in most of the proofs of this section.

Let $k\in [n]$ denote \emph{the} item with the smallest $a_k$ among all 
  $i\in [n]$ with $p_i=a_i$.
By the non-degeneracy assumption, $k$ is unique.
Recall that $t_k=b_k-a_k=b_k-p_k$. We let $S$ denote the set of $i\in [n]$ 
  such that $b_i-p_i=t_k$, so $k\in S$.
By the non-degeneracy assumption again, we have $p_i>a_i$ for all $i\in S-\{k\}$.  
We now construct $\pp'\in P$ as follows: For each $i\in [n]$, set $p_i'=p_i$ if $i\notin S$;
  otherwise set $p_i'=p_i+\epsilon$ for some sufficiently small $\epsilon>0$.
Next we show that $\Revv(\vv,\pp')\ge \Revv(\vv,\pp)$ for all $\vv\in V$. 
Fix a $\vv\in V$. We consider the following three cases:\vspace{0.08cm}
\begin{flushleft}\begin{enumerate}
\item If $\Util(\vv,\pp)=t_k$,
  then $\topU(\vv,\pp)\subseteq S$ by the definition of $S$. When $\epsilon$
  is sufficiently small, we have 
  $$\topU(\vv,\pp')=\topU(\vv,\pp)\ \ \ \ \text{and}\ \ \ \ 
  \Revv(\vv,\pp')=\Revv(\vv,\pp)+\epsilon>\Revv(\vv,\pp).$$
\item If $\Uti(\vv,\pp)=0$ and
  $k\in \topU(\vv,\pp)$, then we have $\topU(\vv,\pp)\cap S=\{k\}$ since
  $b_i>p_i>a_i$ for all other $i\in S$. We claim that $\Revv(\vv,\pp)>p_k$
  in this case. To see this, note that there exists an item $\ell\in [n]$
  such that $p_\ell=a_\ell$ and $p_\ell>p_k$ by our choice of $k$.
As $\Uti(\vv,\pp)=0$, we must have $v_\ell=a_\ell$ and thus, $\ell\in \topU(\vv,\pp)$
  and $\Revv(\vv,\pp)\ge p_\ell$ is not obtained from selling item $k$.
Therefore, we have
$$
\Uti(\vv,\pp')=0,\ \ \ \topU(\vv,\pp')=\topU(\vv,\pp)-\{k\}\ \ \ \text{and}\ \ \ 
\Revv(\vv,\pp')=\Revv(\vv,\pp).
$$
  
\item Finally, if neither of the cases above happens, then we have $\topU(\vv,\pp)\cap S=\emptyset$
  (note that this includes the case when $\topU(\vv,\pp)=\emptyset$).
For this case we have $\topU(\vv,\pp')=\topU(\vv,\pp)$
  and $\Revv(\vv,\pp')=\Revv(\vv,\pp).$\vspace{0.08cm}
\end{enumerate}\end{flushleft}

The lemma then follows because in the second case
  above, we indeed showed that the following valuation vector
  $\vv^*$ in $V$ satisfies $\Revv(\vv^*,\pp')>\Revv(\vv^*,\pp)$: $v_k=b_k$ and $v_i=a_i$ for all $i\ne k$.
\end{proof}

Lemma \ref{secondlem} reduces our search space to $\pp$ such that
  either $\pp=\bb$ or $\pp\in P_k$ for some $k\in [n]$, where we use $P_k$ to denote the set
  of price vectors $\pp\in P$ such that $p_k=a_k$ and 
  $p_i>a_i$ for all other $i\in [n]$.
  
The next lemma further restricts our attention to $\pp\in P_k$ such that
  $p_i\in \{b_i,b_i-t_k\}$ {for all $i\ne k$.

\begin{lemma}\label{popo}
If $\pp\in P_k$ but $p_i\notin \{b_i,b_i-t_k\}$ for some $i\ne k$,
  then we have $\pp\notin \emph{\opt}$.
\end{lemma}
\begin{proof}
Assume for contradiction that 
   $p_\ell\notin \{b_\ell,b_\ell-t_k\}$.
As $\pp\in P_k$, we also have $p_\ell>a_\ell$.
Now we use $S$ to denote the set of all $i\in [n]$ such that $b_i-p_i=b_\ell-p_\ell$.
It is clear that $k\notin S$.
We use $\pp'$ to denote the following new price vector: $p_i'=p_i$ for all
  $i\notin S$, and $p_i'=p_i+\epsilon$ for all $i\in S$, where 
  $\epsilon>0$ is sufficiently small.
We use the same proof strategy to show that $\Revv(\pp')>\Revv(\pp)$. Fix any $\vv\in V$. We have\vspace{0.08cm}
\begin{flushleft}\begin{enumerate}
\item If $\Uti(\vv,\pp)<0$, then clearly $\Uti(\vv,\pp')<0$ as well
  and thus, $\Revv(\vv,\pp')=\Revv(\vv,\pp)=0$.\vspace{-0.1cm}
\item If $\Uti(\vv,\pp)=b_\ell-p_\ell$, then
  $\topU(\vv,\pp)\subseteq S$ by the definition of $S$. When $\epsilon$ is
  sufficiently small,
$$
\topU(\vv,\pp')=\topU(\vv,\pp)\ \ \ \text{and}\ \ \ \
\Revv(\vv,\pp')=\Revv(\vv,\pp)+\epsilon>\Revv(\vv,\pp).
$$
\item If $\Uti(\vv,\pp)\ge 0$ but $\Uti(\vv,\pp)\ne b_\ell-p_\ell$, then 
  it is easy to see that $\topU(\vv,\pp)\cap S=\emptyset$, because $p_i>a_i$ 
  and $b_i-p_i=b_\ell-p_\ell$ for all $i\in S$.
It follows that $\topU(\vv,\pp')=\topU(\vv,\pp)$ {and}
  $\Revv(\vv,\pp')=\Revv(\vv,\pp)$.\vspace{0.08cm}
\end{enumerate}\end{flushleft}
The lemma follows by combining all three cases.
\end{proof}

As suggested by Lemma \ref{popo},
  for each $k\in [n]$, we use $P_k'$ to denote the set of $\pp\in P_k$ such that
  $p_k=a_k$ and $p_i\in \{b_i,b_i-t_k\}$ for all other $i$.
In particular, $p_i$ must be $b_i$ if $t_i<t_k$ ($t_i\ne t_k$, by 
  the non-degeneracy assumption).
The next lemma shows that we only need to consider
  $\pp\in P_k'$ such that $p_i=b_i$ for all $i<k$.

\begin{lemma}
If $\pp\in P_k'$ satisfies $p_\ell=b_\ell-t_k>a_\ell$ 
  for some $\ell<k$, then we have $\pp\notin \emph{\opt}$.
\end{lemma}
\begin{proof}
We construct $\pp'$ from $\pp$ as follows. Let $S$ denote the set of all $i<k$
  such that $p_i=b_i-t_k>a_i$.
By our assumption, $S$ is nonempty.
Then set $p_i'=p_i$ for all $i\notin S$ and
  $p_i'=p_i+\epsilon$ for all $i\in S$, where $\epsilon>0$ is sufficiently small.
Similarly we show that $\Revv(\pp')> \Revv(\pp)$ by  
  considering the following cases:\vspace{0.08cm}
\begin{flushleft}\begin{enumerate}
\item If $\Uti(\vv,\pp)=t_k$ and $\topU(\vv,\pp)\cap S\ne \emptyset$,
  we consider the following cases. 
If $\topU(\vv,\pp)\subseteq S$, then 
$$\topU(\vv,\pp')=\topU(\vv,\pp)\ \ \ \ \text{and}\ \ \ 
\ \Revv(\vv,\pp')=\Revv(\vv,\pp)+\epsilon>\Revv(\vv,\pp).$$
Otherwise, there exists a $j\ge k$ such that
  $j\in \topU(\vv,\pp)$.
This implies that $\Revv(\vv,\pp)\ge p_j=b_j-t_k$ is not obtained from 
  any item in $S$. As a result,
$\topU(\vv,\pp')=\topU(\vv,\pp)-S$ {and} $\Revv(\vv,\pp')=\Revv(\vv,\pp)$.\vspace{-0.05cm}
\item If the case above does not happen, then
  we must have $\topU(\vv,\pp)\cap S=\emptyset$ (this includes the case
  when $\topU(\vv,\pp)=\emptyset$).
As a result, we have $\topU(\vv,\pp')=\topU(\vv,\pp)$ {and} $\Revv(\vv,\pp')=\Revv(\vv,\pp)$.\vspace{0.08cm}
\end{enumerate}\end{flushleft}
The lemma follows by combining the two cases.
\end{proof}

Finally, we use $P_k^*$ for each $k\in [n]$ to denote the set of 
  $\pp\in P$ such that $p_k=a_k$; $p_i=b_i$ for all $i<k$; $p_i=b_i$,
  for all $i>k$ such that $t_i<t_k$; and $p_i\in \{b_i,b_i-t_k\}$,
  for all other $i>k$.
However, $P_k^*$ may still be exponentially large in general. 
Let $T_k$ denote the set of $i>k$ such that $t_i>t_k$.
Given $\pp\in P_k^*$,
  our last lemma below implies that, if $i$ is the smallest index in $T_k$
  such that $p_i=b_i-t_k$, then $p_j=b_j-t_k$ for all $j\in T_k$ larger than $i$;
  otherwise $\pp$ is not optimal.
In other words, $\pp$ has to be monotone in setting $p_j$, $j\in T_k$,
  to be $b_j-t_k$; otherwise $\pp$ is not optimal.
As a result, there are only $O(n^2)$ many price vectors that we need to 
  check, and the best one among them is optimal.
We use $A\subseteq \cup_k P_k^*$ to denote this set of price vectors.

\begin{lemma}\begin{flushleft}
Given $k\in [n]$ and $\pp\in P_k^*$,
  if there exist two indices 
  $c,d\in T_k$ such that $c<d$, $p_c=b_c-t_k$ but $p_d=b_d$,
  then we must have $\pp\notin \emph{\opt}$.
\end{flushleft}\end{lemma}
\begin{proof}
We use $t$ to denote $t_k$ for convenience.
Also we may assume, without loss of generality, that 
  there is no index between $c$ and $d$ in $T_k$; otherwise
  we can use it to replace either $c$ or $d$, depending on its price.

We define two vectors from $\pp$.
First, let $\pp'$ denote the vector obtained from $\pp$ by replacing $p_d=b_d$
  by $p_d'=b_d-t$.
Let $\pp^*$ denote the vector obtained from $\pp$ by replacing $p_c=b_c-t$
  by $p_c^*=b_c$.
In other words, the $c$th and $d$th entries of $\pp,\pp',\pp^*$ are
  $(b_c-t,b_d),(b_c-t,b_d-t),(b_c,b_d)$, respectively, while all other $n-2$ 
    entries are the same.
Our plan is to show that if $\Revv(\pp)\ge \Revv(\pp')$, then $\Revv(\pp^*)>\Revv(\pp)$.
This implies that $\pp$ cannot be optimal and the lemma follows.

We need some notation. 
Let $V'$ denote the projection of $V$ onto 
  all but the $c$th and $d$th coordinates:
$$V'=\times_{i\in [n]-\{c,d\}} V_i.$$ 
We use $[n]-\{c,d\}$ to index entries of vectors $\uu$ in $V'$.
Let $U\subseteq V'$ denote the set of vectors $\uu\in V'$ such that
  $u_i-p_i<t$ for all $i>d$.
(This just means that for each $i\in T_k$, if $i>d$ and $p_i=b_i-t$, then $u_i=a_i$.)
Given $\uu\in V'$, $v_c\in \{a_c,b_c\}$
  and $v_d\in \{a_d,b_d\}$, we use $(\uu,v_c,v_d)$ to denote
  a $n$-dimensional price vector in $V$.
Now we compare the expected revenue $\Revv(\pp)$, $\Revv(\pp')$ and $\Revv(\pp^*)$. 

First, we claim that,
  if $\vv=(\uu,v_c,v_d)\in V$ but $\uu\notin U$, then we have
  $\Revv(\vv,\pp)=\Revv(\vv,\pp')=\Revv(\vv,\pp^*).$
This is simply because there exists an item $i>d$
  such that $v_i-p_i=t$, so it always
  dominates both items $c$ and $d$. As a result,
  the difference among $\pp,\pp'$ and $\pp^*$ no longer matters.
Second, it is easy to show that for any $\vv=(\uu,a_c,a_d)\in V$, then 
  $\Revv(\vv,\pp)=\Revv(\vv,\pp')=\Revv(\vv,\pp^*)$
  as the utility from $c$ and $d$ are negative.
  
Now we consider a vector $\vv=(\uu,v_c,v_d)\in V$ such that
  $\uu\in U$ and $(v_c,v_d)$ is either $(a_c,b_d)$, $(b_c,a_d)$, or $(b_c,b_d)$.
For convenience, for each $\uu\in U$ we use $\uu_1^+$ to denote $(\uu,a_c,b_d)$;
  $\uu_2^+$ to denote $(\uu,b_c,a_d)$; and $\uu_3^+$ to denote $(\uu,b_c,b_d)$.
By the definition of $U$, we have the following simple cases: 
\begin{enumerate}
\item For $\pp$, we have $\Revv(\uu_2^+,\pp)=b_c-t$ and $\Revv(\uu_3^+,\pp)=b_c-t$;\vspace{-0.14cm}
\item For $\pp'$, we have $\Revv(\uu_1^+,\pp')=b_d-t$, $\Revv(\uu_2^+,\pp')=b_c-t$ and
  $\Revv(\uu_3^+,\pp')=b_d-t$.
\end{enumerate}
We need the following equation:
\begin{equation}\label{iii}
\Revv(\uu_1^+,\pp)=\Revv(\uu_1^+,\pp^*)=\Revv(\uu_3^+,\pp^*)
\end{equation}
as well as the following two inequalities:
\begin{equation}\label{imp2}
\Revv(\uu_1^+,\pp^*)-(b_d-b_c)\le  \Revv(\uu_2^+,\pp^*)\le \Revv(\uu_1^+,\pp^*)
\end{equation}

Given a $\vv\in V$, recall that $\Pr[\vv]$ denotes the probability 
  of the valuation vector being $\vv$.
Given a $\uu\in U$, we also use $\Pr[\uu]$ to denote the probability of 
  the $n-2$ items, except items $c$ and $d$, taking values $\uu$.
Let $$h_1=(1-q_c)q_d,\ \ \ \ h_2=q_c(1-q_d)\ \ \ \ \text{and}\ \ \ \ h_3=q_cq_d.$$
Clearly we have 
  $h_1,h_2,h_3>0$ and 
$\Pr[\uu_i^+]=\Pr[\uu]\cdot h_i$, {for all $\uu\in U$ and $i\in [3]$.}
 
In order to compare $\Revv(\pp)$, $\Revv(\pp')$ and $\Revv(\pp^*)$, we only need to compare
  the following three sums:
$$
\sum_{i\in [3]}\sum_{\uu\in U}\Pr[\uu_i^+]\cdot \Revv(\uu_i^+,\pp),\ \ \ \ 
\sum_{i\in [3]}\sum_{\uu\in U}\Pr[\uu_i^+]\cdot \Revv(\uu_i^+,\pp')\ \ \ \ \text{and}\ \ \ \
\sum_{i\in [3]}\sum_{\uu\in U}\Pr[\uu_i^+]\cdot \Revv(\uu_i^+,\pp^*).
$$
For the first sum, we can rewrite it as (here all sums are over $\uu\in U$):
\begin{equation}\label{sum1}
h_1\cdot \sum_{\uu} \Pr[\uu] \cdot \Revv(\uu_1^+,\pp)+h_2\cdot \sum_{\uu}\Pr[\uu] \cdot (b_c-t)
  +h_3\cdot \sum_\uu \Pr[\uu] \cdot (b_c-t),
\end{equation}
while the sum for $\Revv(\pp')$ is the following:
\begin{equation}\label{sum2}
h_1\cdot\sum_{\uu} \Pr[\uu] \cdot (b_d-t)+h_2\cdot \sum_\uu\Pr[\uu] \cdot (b_c-t)
  +h_3\cdot \sum_\uu \Pr[\uu] \cdot (b_d-t).
\end{equation}
Since $c<d$ and $b_c<b_d$, $\Revv(\pp)\ge \Revv(\pp')$ would imply that
\begin{equation}\label{imp}
\sum_\uu \Pr[\uu]\cdot \Revv(\uu_1^+,\pp)>\sum_\uu \Pr[\uu]\cdot (b_d-t).
\end{equation}

On the other hand, we can also rewrite the sum for $\Revv(\pp^*)$ as
\begin{equation}\label{sum3}
h_1\cdot\sum_{\uu} \Pr[\uu] \cdot \Revv(\uu_1^+,\pp^*)+h_2\cdot\sum_\uu\Pr[\uu] \cdot \Revv(\uu_2^+,\pp^*)
  +h_3\cdot \sum_\uu \Pr[\uu] \cdot \Revv(\uu_3^+,\pp^*).
\end{equation}
The first sum in (\ref{sum3}) is the same as that of (\ref{sum1}). For the second sum, 
  from (\ref{imp2}), (\ref{iii}) and (\ref{imp}) we have
\begin{align*}
\sum_\uu\Pr[\uu]\cdot \Revv(\uu_2^+,\pp^*)
\hspace{0.04cm}&\ge\hspace{0.04cm} \sum_\uu\Pr[\uu]\cdot 
  \Big(\Revv(\uu_1^+,\pp)-(b_d-b_c)\Big)\\[0.8ex] &>\hspace{0.04cm} \sum_\uu \Pr[\uu]\cdot 
  \Big(b_d-t-(b_d-b_c)\Big)=\sum_\uu \Pr[\uu]\cdot (b_c-t).
\end{align*} 
The third sum in (\ref{sum3}) is also strictly
  larger than that of (\ref{sum1}) as
$\Revv(\uu_3^+,\pp^*)=\Revv(\uu_1^+,\pp^*)\ge \Revv(\uu_2^+,\pp^*)$
while the second and third sums in (\ref{sum1}) are the same, ignoring $h_2$ and $h_3$.
Thus, $\Revv(\pp^*)>\Revv(\pp)$. 
\end{proof}

\subsection{General Case} \label{ssec:general-case}

Now we deal with the general case.
Let $I$ denote an input instance with $n$ items, in which $|V_i|\le 2$ for all $i$.
For each $i\in [n]$, either $V_i=\{a_i,b_i\}$ where $b_i>a_i\ge 0$,
  or $V_i=\{b_i\}$, where $b_i\ge 0$.
We let $D\subseteq [n]$ denote the set of $i\in [n]$ such that $|V_i|=2$.
For each item $i\in D$, we use $q_i:0<q_i<1$ to
  denote the probability of its value being $b_i$. 
Each item $i\notin D$ has value $b_i$ with probability $1$.
As permuting the items does not affect the maximum
  expected revenue, we may assume without loss of generality that
  $b_1\le b_2\le \cdots\le b_n$.

The idea is to perturb $I$ (symbolically), so that the new instances 
  satisfy all conditions described at the beginning of the section,
  which we know how to solve efficiently.
For this purpose, we define a new $n$-item 
  instance $I_\epsilon$ from $I$ for any $\epsilon>0$: 
For each $i\in D$, the support of
  item $i$ is $V_{i,\epsilon}=\{a_i+i\epsilon,b_i+2i\epsilon\}$, and
  for each $i\notin D$, the support of item $i$ is $V_{i,\epsilon}=\{b_i+i\epsilon,b_i+2i\epsilon\}$.
For each $i\in D$, the probability of the value being $b_i+2i\epsilon$
  is still set to be $q_i$, while for each $i\notin D$, the probability of the value
  being $b_i+2i\epsilon$ is set to be $1/2$.  
In the rest of the section, we use $\Revv(\pp)$ and $\Revv(\vv,\pp)$
  to denote the revenue with respect to $I$, and use $\Revv_\epsilon(\pp)$ and $\Revv_\epsilon(\vv,\pp)$
  to denote the revenue with respect to $I_\epsilon$.
Let $V_\epsilon=\times_{i=1}^n V_{i,\epsilon}$.
Let $\rho$ denote the following map from $V_{\epsilon}$ to $V$:
  $\rho$ maps $\uu\in V_\epsilon$ to $\vv\in V$, where
  1) $v_i=b_i$ when $i\notin D$; 2) $v_i=a_i$ if $u_i=a_i+i\epsilon$ 
  and $v_i=b_i$ if $u_i=b_i+2i\epsilon$ when $i\in D$.

It is easy to verify that, when $\epsilon>0$ is sufficiently small,
  the new instance $I_\epsilon$ satisfies all conditions given at the beginning 
  of the section, including the non-degeneracy assumption.
Moreover, we show that

\begin{lemma}\label{haha}
The limit of $\max_{\pp} \Revv_\epsilon(\pp)$ exists as $\epsilon\rightarrow0$,
  and can be computed in polynomial time.
\end{lemma}
\begin{proof}
Since $I_\epsilon$ satisfies all the conditions, we know there is a 
  set of $O(n^2)$ price vectors, denote by $A_\epsilon$ for $I_\epsilon$,
  such that the best vector in $A_\epsilon$ is optimal for $I_\epsilon$
  and achieves $\max_\pp \Revv_\epsilon(\pp)$.

Furthermore, from the construction of $A_\epsilon$, we know that
  every vector $\pp_\epsilon$ in $A_\epsilon$ has an explicit expression
  in $\epsilon$: each entry of $\pp_\epsilon$ is indeed an affine linear 
  function of $\epsilon$. 
As a result, the limit of $\Revv_\epsilon(\pp_\epsilon)$ as
  $\epsilon$ approaches $0$ exists and can be computed efficiently.
Since $\lim_{\epsilon\rightarrow 0} \left(\max_\pp \Revv_\epsilon(\pp)\right)$
  is just the maximum of these $O(n^2)$ limits, it also exists and can
  be computed in polynomial time in the input size of $I$.
\end{proof}

Finally, the next two lemmas show that 
  this limit is exactly the maximum expected revenue of $I$.

\begin{lemma}
$\max_\pp \Revv(\pp)\le \lim_{\epsilon\rightarrow 0} 
  \big(\max_\pp \Revv_\epsilon(\pp)\big)$.
\end{lemma}
\begin{proof}
Let $\pp^*$ denote an optimal price vector of $I$.
It suffices to show that, when $\epsilon$ is sufficiently small,
\begin{eqnarray}\label{oo}
&\max_\pp\Revv_\epsilon(\pp)\ge \Revv(\pp^*)-4n^2\epsilon.&
\end{eqnarray}

The proof is similar to that of Lemma \ref{tie-breaking-ha}.
Let $\pp'$ denote the vector in which 
  $p_i'=\max\left(0,p^*_i-4r_in\epsilon\right)$, where $r_i$ is the rank of $p_i^*$
  among $\{p_1^*,\ldots,p_n^*\}$ sorted in the increasing order (when there are ties, items with
  lower index are ranked higher).  
We claim that, when $\epsilon>0$ is sufficiently small,
\begin{equation}\label{goal1}
\Revv_\epsilon(\uu,\pp')\ge \Revv(\rho(\uu),\pp^*)-4n^2\epsilon,
  \ \ \ \text{for any $\uu\in V_\epsilon$,}
\end{equation}
from which we get $\Revv_\epsilon(\pp')\ge \Revv(\pp^*)-4n^2\epsilon$
  and (\ref{oo}) follows.

To prove (\ref{goal1}) we fix a $\uu\in V_\epsilon$ and let $\vv=\rho(\uu)\in V$.
(\ref{goal1}) holds trivially if $\Revv(\vv,\pp^*)=0$.
Assume that $\Revv(\vv,\pp^*)>0$, and let $k$ denote the item selected in $I$ on $(\vv,\pp^*)$.
(\ref{goal1}) also holds trivially if $p_k^*<4n^2\epsilon$, so without loss 
  of generality, we assume that
  $p_k\ge 4n^2\epsilon$.
For any other item $j\in [n]$, we compare the utilities of 
  items $k$ and $j$ in $I_\epsilon$ on $(\uu,\pp')$.
We claim that
\begin{equation}\label{goal2}
u_k-p_k'>u_j-p_j'
\end{equation}
because 1) if $v_k-p_k^*>v_j-p_j^*$, then (\ref{goal2}) holds when $\epsilon$
  is sufficiently small;
  2) if $v_k-p^*_k=v_j-p^*_j$ and $p^*_k>p^*_j$, then (\ref{goal2}) holds because
  $p^*_k-p_k'-(p^*_j-p_j')\ge 4n\epsilon >(v_k-u_k)+(u_j-v_j)$;
  3) finally, the case when $v_k-p^*_k=v_j-p^*_j$, $p_k=p_j$ and $k<j$ follows similarly
  from $r_k>r_j$.
Therefore, $k$ remains to be the item being selected in $I_\epsilon$ on $(\uu,\pp')$.
(\ref{goal1}) then follows from the fact that $p_k'\ge p_k^*-4n^2\epsilon$ by definition.
\end{proof}

\begin{lemma}
$\max_\pp \Revv(\pp)\ge \lim_{\epsilon\rightarrow 0} 
  \big(\max_\pp \Revv_\epsilon(\pp)\big)$.
\end{lemma}  
\begin{proof}
From the proof of Lemma \ref{haha}, there is a 
  price vector $\pp_\epsilon\in A_\epsilon$ in which every entry
  is an affine linear function of $\epsilon$, such that (as the cardinality of
  $|A_\epsilon|$ is bounded from above by $O(n^2)$)
$$
\lim_{\epsilon\rightarrow 0} \left(\max_\pp \Revv_\epsilon(\pp)\right)
= \lim_{\epsilon\rightarrow 0} \Revv_\epsilon(\pp_\epsilon).
$$
Let $\widetilde{\pp}\in \mathbb{R}_+^n$ denote the limit of $\pp_\epsilon$, by simply removing
  all the $\epsilon$'s in the affine linear functions.
Moreover, we note that $|\hspace{0.03cm}\widetilde{p}_i-p_{\epsilon,i}\hspace{0.03cm}|=O(n\epsilon)$
  by the construction of $A_\epsilon$, where we use $p_{\epsilon,i}$ to denote
  the $i$th entry of $\pp_\epsilon$.
  
Next, let $\qq_\epsilon$ denote the vector in which
  the $i$th entry $q_{\epsilon,i}= \max\hspace{0.05cm}(0,\widetilde{p}_i-r_in^2\epsilon)$ 
  for all $i\in [n]$, where $r_i$ is the rank of $\widetilde{p}_i$ among entries of $\widetilde{\pp}$
  sorted in increasing order (again, when there are ties, items with lower index are ranked higher).
To prove the lemma, it suffices to show that, when $\epsilon$ is sufficiently small,
$$
\Revv(\qq_\epsilon)\ge \Revv_\epsilon(\pp_\epsilon)-O(n^3\epsilon).
$$
To this end, we show that for any vector $\uu\in V_\epsilon$ with $\vv=\rho(\uu)$,
\begin{equation}\label{goal100}
\Revv(\vv,\qq_\epsilon)\ge \Revv_\epsilon(\uu, \pp_\epsilon)-O(n^3\epsilon).
\end{equation}

Finally we prove (\ref{goal100}).
First, we note that if $\Util(\vv,\widetilde{\pp})<0$, then
  $\Revv(\vv,\qq_\epsilon)=\Revv_\epsilon(\uu,\pp_\epsilon)=0$ when
  $\epsilon>0$ is sufficiently small (as $\uu$ approaches $\vv$ and $\pp_\epsilon$, $\qq_\epsilon$ 
  approach $\widetilde{\pp}$).
Otherwise, we have $\Util(\vv,\qq_\epsilon)
  >\Util(\vv,\widetilde{\pp})\ge 0$ and we use $k$ to denote the item 
  selected in $I$ on $(\vv,\qq_\epsilon)$.
To violate (\ref{goal100}), the item selected in $I_\epsilon$ on $(\uu,\pp_\epsilon)$
  must be an item $\ell$ different from $k$ satisfying $\widetilde{p}_\ell>\widetilde{p}_k$.
Below we show that this cannot happen.
Consider all the cases:
  1) if $v_k-\widetilde{p}_k<v_\ell-\widetilde{p}_\ell$, we get a contradiction since item $k$
    is dominated by $\ell$ in $I$ on $(\vv,\qq_\epsilon)$ when $\epsilon$ is sufficiently small;
  2) if $v_k-\widetilde{p}_k>v_\ell-\widetilde{p}_\ell$, we get a contradiction with $\ell$
    being selected in $I_\epsilon$ on $(\uu,\pp_\epsilon)$ when $\epsilon$ is sufficiently small;
  3) if $v_k-\widetilde{p}_k=v_\ell-\widetilde{p}_\ell$ and $\widetilde{p}_\ell>\widetilde{p}_k$, we conclude
    that $v_k-q_{\epsilon,k}<v_\ell-q_{\epsilon,\ell}$, contradicting again with
    $k$ being selected in $I$ on $(\vv,\qq_\epsilon)$.
(\ref{goal100}) follows by combining all these cases. 
\end{proof}








\section{NP--hardness for support size 3}\label{nphard-sec}

In this section, 
  we give a polynomial-time reduction from \textsc{Partition} 
  to \itemp for distributions with support (at most) $3$.
Recall that in the \textsc{Partition} problem~\cite{GJ79} we are given a set $C = \{c_1, \ldots, c_n\}$ of $n$ positive integers  and wish
to determine whether it is possible to partition $C$ into two subsets with equal sum.
We may assume without loss of generality that 
$c_1 = \max\left(c_1,\ldots, c_n\right)$.

Given an instance  
  of \textsc{Partition}, we construct an instance of \itemp as follows.~We 
have $n$ items. Each item $i \in [n]$ 
  can take $3$ possible integer values $0, a, b$, where $b>a>0$, 
i.e., $V_i=\{0,a,b\}$ for all $i\in [n]$.
Let $q_i = \Pr[v_i = b]$ and $r_i = \Pr[v_i = a]$.
We set $q_i= c_i/M$ where $M=2^n c_1^3$ and $$r_i= \frac{b-a}{a(1-t_i)}\cdot q_i,\ \ \ \ \text{where 
  \ $t_i=\frac{b}{2a}\cdot \sum_{j \neq i, j \in [n]} q_j$.}$$
The two parameters $a$ and $b$ should be thought of as universal constants
(independent of the given instance of \textsc{Partition}) throughout the proof.  
We will eventually set these constants to be $a=1$, $b=3$ (this choice is not necessary, there is flexibility in our proof and indeed any values with $b>2a$ will work). 
However, for the sake of the presentation, we will keep $a,b$ 
  as generic parameters for most of the calculations till the end.

Note that the definition of $r_i$ implies that
\begin{equation}\label{eqeqeq}
bq_i=a(q_i+r_i) - ar_it_i.
\end{equation}
Let $N=2^n c_1^2$. Then we have $q_i,r_i = O(1/N)$ and $t_i=O(n/N)$ for all $i$.
{\em Thus, each distribution~assigns most of its probability mass to the point $0$.}
This is a crucial property which allows us to get a handle~on the optimal revenue. 
For an arbitrary general instance of the pricing problem, the 
  expected revenue
  is a highly complex nonlinear function. The fact that most of the probability mass in our
construction is concentrated at 0 implies that valuation vectors with 
many nonzero entries contribute very little to the expected revenue. 
As we will argue, the revenue is approximated well by its 1st and 2nd order terms with respect to $\mathrm{poly}(n)/N$, which essentially corresponds to the contribution of all valuations in which at most two items have nonzero value. 
The probabilities $q_i, r_i$ are chosen carefully so that
the optimization of the expected revenue amounts to a quadratic optimization problem, which achieves its maximum possible value when the given set $C$ of integers has a partition into two parts with equal sums.


Our main claim is that, for an appropriate value $t^{\ast}$, there exists a price vector with expected revenue at least $t^{\ast}$
if and only if there exists a solution to the original instance of the Partition problem.

Before we proceed with the proof, we will need some notation. 
For $T_1, T_2, \eps \in \mathbb{R}_{+}$ we write 
$T_1= T_2 \pm \eps$ to denote that $|T_1 - T_2| \le \eps$.

Note that, as both the $q_i$'s and the $t_i$'s are very small positive quantities, we have that $r_i \approx (b-a)q_i/a$.
Formally, with the above notation we can write
\begin{equation}\label{usefuleq1}
r_i= \frac{b-a}{a(1-t_i)}\cdot q_i=\frac{b-a}{a}\cdot q_i\pm 2\frac{b-a}{a}\cdot q_it_i
  =\frac{b-a}{a}\cdot q_i \pm O(n/N^2).
\end{equation}

Lemma \ref{searchspace} and Corollary \ref{integerprice} imply that a revenue maximizing  
price vector can be assumed to have~non-negative integer coefficients of magnitude at most $b$.
The following lemma establishes the stronger statement that, for our particular instance, an 
  optimal price vector $\pp$ can be assumed to 
have each $p_i$ in the set $\{a, b\}$.
  
\begin{lemma} \label{lem:prices-in-support}
There is an optimal price vector $\pp \in \{a,b\}^n$.
\end{lemma}  
\begin{proof}
By Lemma \ref{searchspace} and Corollary \ref{integerprice}, there is
an optimal price vector with integer coordinates in $ [0: b]$.
Let $\pp$ be any (integer) vector in $ [0: b]^n$ that has at least one coordinate
$p_j \not\in \{a,b\}$.
We will show below that $\Revv(\pp) < \Revv(\bb)$, where $\bb$ denotes the all-$b$ vector,
and hence $\pp$ is not optimal.


Consider an index $i \in [n]$ with $p_i>0$. 
The probability the buyer selects item $i$ is bounded from above by $\Pr[v_i \ge p_i]$, 
the probability that item $i$ has value at least $p_i$, and is bounded from below by 
$$
\Pr\big[ v_i \ge p_i \big] \cdot \prod_{j\ne i, j \in [n]} (1-q_i-r_i)
\ge \Pr\big[ v_i \ge p_i \big]\cdot \left( 1-O(n/N) \right).
$$
Note that the second term in the LHS above is the probability that all items other than $i$ have value $0$ and the inequality 
uses the fact that $q_i, r_i = O(1/N)$.
Applying these two bounds on $\pp$ and $\bb$ we obtain
$$
\Revv(\bb)\ge \sum_{i\in [n]} q_i\left(1-O(n/N)\right)\cdot b\ \ \ \ \text{and}\ \ \ \ 
\Revv(\pp)\le \sum_{i:p_i>0} \Pr\big[ v_i \ge p_i \big]\cdot p_i.
$$
So $\Revv(\bb)\ge (\sum_{i\in [n]} q_ib) - O(n^2/N^2)$.
Regarding $\Revv(\pp)$, we consider the following three cases. For $i \in [n]$ with $p_i=b$, the probability that $v_i \ge p_i$ is $q_i$
 and the contribution of such an item to the second sum is $q_ib$. 
Similarly, for $i \in [n]$ with $p_i=a$, the probability that $v_i \ge p_i$ is $q_i+r_i$ and the contribution to the sum 
is
$$
(q_i+r_i)a\le q_ib+O(n/N^2),
$$
where the inequality follows from (\ref{usefuleq1}).
Finally, we consider an item $i\in [n]$ with $p_i\notin \{a,b\}$.
If $a<p_i<b$ then the contribution is $q_ip_i$,
which is at most $q_i(b-1)=q_i b - q_i$, since $p_i$ is integer.
If $p_i<a$, then the contribution is $(q_i+r_i)p_i$,
which is at most $(q_i+r_i)(a-1)=q_ib + ar_it_i-q_i -r_i =q_ib -q_i -r_i(1-at_i)$. 
In both cases, the contribution to the sum is at most 
$$q_ib-q_i \leq q_ib-(1/M).$$
Note that the definition of $M$ and $N$ implies that $1/M \gg n^2/N^2$.
Because there exists at least one $j$ with $p_j\notin \{a,b\}$, 
it follows that
$\Revv(\pp) < \Revv(\bb)$ which completes the proof of the lemma.
\end{proof}

As a result, to maximize the expected revenue it suffices to consider price vectors in $\{a,b\}^n$.
Given~any price-vector $\pp \in \{a, b\}^n$, we 
  let $S = S(\pp)= \left\{ i \in [n] : p_i =a \right\}$ and  $T = T(\pp) = \left\{ i \in [n] : p_i =b \right\}$. 
The main idea of the proof is to establish an appropriate 
  \emph{quadratic form approximation}
  to the expected revenue $\Revv(\pp)$ that is sufficiently accurate
for the purposes of our reduction.  
 
\medskip
 
\noindent {\bf Approximating the Revenue.} We appropriately partition the 
  valuation space $V$ into three events that yield positive revenue. 
We then approximate the probability of each and its  contribution to the expected revenue up to, and including, 2nd order terms, 
i.e., terms of order $O(\text{poly}(n)/N^2)$, and we ignore 3rd order terms, i.e., terms of order $O(\epsilon)$ where
 $\epsilon=n^3/N^3$. 
 
In particular, we consider the following disjoint events:\vspace{0.1cm}
\begin{flushleft}\begin{itemize}
\item {\em First Event}: $E_1  = \{\vv\in V \mid \exists\hspace{0.05cm}i \in S: v_i=b \}$.\newline 
Note that for any $\vv \in E_1$ we have $\Revv(\vv, \pp) = a$.
The probability of this event is 
$$ \Pr[E_1] = 1-\prod_{i \in S}(1-q_i)= \sum_{i \in S}q_i - \sum_{i \neq j \in S} q_iq_j
  \pm O(\epsilon).$$
\item {\em Second Event}: $E_2 = \overline{E_1} \cap \left\{ \vv\in V \mid \exists\hspace{0.05cm} i \in S: v_i=a 
  \textrm{ and } \forall\hspace{0.05cm} i \in T: v_i \in \{0,a\} \right\}$.
\newline Note that for any $\vv \in E_2$ we have $\Revv(\vv, \pp) = a$.
  The probability of this event is 
\begin{align*}
 \Pr[E_2] = &\prod_{j\in T}(1-q_j) \left[\hspace{0.06cm} \prod_{i\in S}(1-q_i) - \prod_{i\in S} (1-q_i-r_i) \right]
\end{align*}
Using the elementary identities
\begin{align*}
\prod_{j\in T}(1-q_j)&=1-\sum_{j\in T} q_j+\sum_{i\ne j\in T}q_iq_j\pm O(\epsilon)\\[0.36ex]
\prod_{i\in S}(1-q_i)&=1-\sum_{i\in S}q_i +\sum_{i\ne j\in S} q_iq_j\pm O(\epsilon)\\[0.36ex]
\prod_{i\in S}(1-q_i-r_i)&= 1-\sum_{i\in S}(q_i+r_i)+\sum_{i\ne j\in S}(q_i+r_i)(q_j+r_j)
  \pm O(\epsilon),
\end{align*}
we can write
\begin{align*}
\hspace{-1cm}\Pr[E_2]&=\left[ 1- \sum_{j \in T}q_j + \sum_{i \neq j \in T} q_iq_j\pm O(\epsilon) \right]
  \hspace{-0.06cm} \cdot \hspace{-0.06cm}
\left[ \sum_{i \in S}r_i + \sum_{i\neq j \in S} q_iq_j - \sum_{i\neq j \in S} (q_i+r_i)(q_j+r_j) 
  \pm O(\epsilon)\right]\\[0.7ex]
&= \sum_{i \in S}r_i - \sum_{i\in S}r_i \sum_{j\in T}q_j + \sum_{i\neq j \in S} q_iq_j-
\sum_{i\neq j \in S} (q_i+r_i)(q_j+r_j)\pm O(\epsilon).
\end{align*}
\item {\em Third Event}: $E_3 = \overline{E_1} \cap \left\{ \vv\in V \mid 
  \exists\hspace{0.05cm} i \in T: v_i=b \right\}.$
\newline Note that for any $\vv \in E_3$ we have $\Revv(\vv, \pp) = b$.
The probability of this event is 
\begin{align*}
\Pr[E_3] &= \prod_{i \in S}(1-q_i) \left[1-\prod_{j\in T} (1-q_j)\right] \\
&= \left( 1- \sum_{i \in S}q_i + \sum_{i \neq j \in S} q_iq_j\pm O(\epsilon) \right)
\left( \sum_{j \in T}q_j - \sum_{i \neq j \in T} q_iq_j \pm O(\epsilon)\right) \\[0.6ex]
&= \sum_{j \in T}q_j - \sum_{i \neq j \in T}q_iq_j - \sum_{i\in S}q_i \sum_{j \in T} q_j \pm O(\epsilon).
\end{align*}
\end{itemize}
\end{flushleft}
Therefore, for the expected revenue $\Revv(\pp)$ we have:
\begin{align*}
\Revv(\pp) &=\big( \Pr[E_1]+ \Pr[E_2] \big) \cdot a + \Pr[E_3]\cdot b\\
                    &= a \cdot \hspace{-0.06cm}\left( \sum_{i \in S} (q_i+r_i) - \sum_{i\neq j \in S} (q_i+r_i)(q_j+r_j) - \sum_{i\in S}r_i \sum_{j\in T}q_j \right)\hspace{-0.06cm}\\
&\hspace{0.6cm}+ b \cdot \hspace{-0.06cm}\left( \sum_{j \in T}q_j - \sum_{i \neq j \in T}q_iq_j - \sum_{i\in S}q_i \sum_{j \in T} q_j \right) \pm O(\eps).
\end{align*}
Using (\ref{eqeqeq}) it follows that the first order term of the revenue is
$$b \sum_{j \in T}q_j + a \sum_{i \in S} (q_i+r_i) = b \sum_{j \in [n]}q_j + 
\sum_{i \in S} \big(a(q_i+r_i)-bq_i\big) =  b \sum_{j \in [n]}q_j + \sum_{i \in S} (ar_it_i).$$
Observe that the first term $b \sum_{j \in [n]}q_j$ in the above expression is a constant $L_1$, independent of the pricing 
(i.e., the partition of the items into $S$ and $T$).

In the second order term, we can rewrite the expression
$a \sum_{i\neq j \in S} (q_i+r_i)(q_j+r_j)$ as
\begin{align*}
&\frac{1}{2}\cdot\sum_{i \in S} (q_i+r_i) \sum_{ j \in S,\hspace{0.05cm} j\neq i} a(q_j+r_j)\\[0.36ex]
&=\frac{1}{2}\cdot\sum_{i \in S} (q_i+r_i) \sum_{ j \in S,\hspace{0.05cm} j \neq i} (bq_j+ar_jt_j)\\
&=\frac{b}{2}\cdot \sum_{i \in S} q_i \sum_{ j \in S,\hspace{0.05cm} j \neq i} q_j + 
\frac{b}{2}\cdot\sum_{i \in S} r_i \sum_{ j \in S,\hspace{0.05cm} j \neq i} q_j + 
\frac{1}{2}\cdot\sum_{i \in S} (q_i+r_i) \sum_{ j \in S,\hspace{0.05cm} j \neq i} ar_jt_j \\[0.36ex]
&= b \sum_{i\neq j \in S} q_i q_j + \frac{b}{2}\sum_{i \in S} r_i \sum_{ j \in S,\hspace{0.05cm} j \neq i} q_j
\pm O(\epsilon)
\end{align*}
where in the first expression above, the double summation is multiplied by $1/2$
  because each unordered pair $i\neq j \in S$ is included twice.
Thus, the second order term of the expected revenue $\Revv(\pp)$ is
\begin{align*}&- a\sum_{i\neq j \in S} (q_i+r_i)(q_j+r_j) - a\sum_{i\in S}r_i \sum_{j\in T}q_j 
-b\sum_{i \neq j \in T}q_iq_j - b\sum_{i\in S}q_i \sum_{j \in T} q_j\\[0.1ex]
&=-b \sum_{i\neq j \in S} q_i q_j - \frac{b}{2}\sum_{i \in S} r_i 
  \sum_{ j \in S,\hspace{0.05cm} j \neq i} q_j - a\sum_{i\in S}r_i \sum_{j\in T}q_j
-b\sum_{i \neq j \in T}q_iq_j - b\sum_{i\in S}q_i \sum_{j \in T} q_j\pm O(\epsilon)\\[0.1ex]
&=-b \sum_{i\neq j \in [n]} q_i q_j - \frac{b}{2}\sum_{i \in S} r_i \sum_{ j \in S,
  \hspace{0.05cm}j \neq i} q_j - a\sum_{i\in S}r_i \sum_{j\in T}q_j \pm O(\epsilon)
 \end{align*}
The first term in the last expression is a constant 
  $L_2$ independent of the pricing.
As a result, we can rewrite the second order term as follows:
$$L_2 - \frac{b}{2}\sum_{i \in S} r_i \sum_{ j \in S,\hspace{0.05cm} j \neq i} q_j - 
  a\sum_{i\in S}r_i \sum_{j\in T}q_j   \pm O(\epsilon)
  = L_2 - \sum_{i \in S} r_i \left( \frac{b}{2} 
  \sum_{j \in S,\hspace{0.05cm} j\neq i} q_j+ a \sum_{j \in T} q_j \right) \pm O(\epsilon).$$
Summing with the fist order term and letting $L=L_1+L_2$, we have:
\begin{align*}
\Revv(\pp) &=  L + \sum_{i \in S} r_i \left( at_i -\frac{b}{2} 
  \sum _{j \in S,\hspace{0.05cm} j\neq i} q_j-a \sum_{j \in T} q_j \right) \pm O(\epsilon)\\[0.16ex]&= L + \sum_{i \in S} r_i 
  \left( \frac{b}{2} \sum _{j \neq i}q_j -\frac{b}{2} \sum _{j \in S,\hspace{0.05cm} j\neq i} q_j-a \sum_{j \in T} q_j \right) 
  \pm O(\epsilon) \\[0.66ex]
&=L + \sum_{i \in S} r_i \cdot \left(\frac{b}{2}-a\right) \sum_{j \in T} q_j \pm O(\epsilon)  \\[0.4ex]
&=L + \frac{b-a}{a}\cdot\left(\frac{b}{2}-a\right)\cdot \frac{1}{M^2}\cdot \sum_{i \in S} c_i\cdot  \sum_{j \in T} c_j
\pm O(\epsilon).
\end{align*}
Now setting $a=1, b=3$ in the previous expression, we have that  for any $\pp\in \{a,b\}^n$,
\begin{equation} \label{eqn:final}
\Revv(\pp) = L + \frac{1}{M^2} \left( \sum_{i \in S} c_i \right) \cdot \left(\sum_{j \in T} c_j \right) \pm O(\eps).
\end{equation}

At this point, we observe that the sum of the two factors $\sum_{i \in S} c_i, \sum_{j \in T} c_j$ in (\ref{eqn:final}) is a constant
(independent of the partition). Thus, their product is maximized when they are equal. 
Because $\eps = o(1/M^2)$, it follows that
the revenue is maximized when the product of the two factors is maximized.
In particular, if there exists a partition of the set $C = \{c_1,\ldots,c_n\}$ 
into two sets with equal sums $H= (\sum_{i\in [n]} c_i)/2$,
then the corresponding partition of the indices into the sets $S$ and $T$
yields revenue $L+ \frac{1}{M^2} \cdot H^2 \pm O(\eps)$.
On the other hand, if there is no such equipartition of the set $C$, then
for any partition of the indices, the revenue will be at most
$L + \frac{1}{M^2} (H+1)(H-1) \pm O(\eps) = L + \frac{1}{M^2} (H^2-1) \pm O(\eps)$.
Since $\eps = o(1/M^2)$ it follows that there exists a partition of the set $C = \{c_1,\ldots,c_n\}$ into two sets with equal sums if and only if
there exists a price vector $\pp\in \{a,b\}^n$ with 
  $\Revv(\pp)\ge t^{\ast} = L+ \frac{1}{M^2} (H^2-\frac{1}{2})$. This completes the proof.

\medskip
\noindent
{\bf Remark.} In the above construction, the support $\{0,a,b\}$ of the distributions
includes the value 0 (which in fact has most of the probability mass).
It is easy to modify the construction, if desired, so that the support
contains only positive values: shift all the values of the distributions up by 1
(thus, the supports now become $V_i=\{1,2,4\}$) and add an additional $(n+1)$-th item
which has value 1 with probability 1.
This transforma\-tion increases the expected revenue by 1.
It is easy to see that an optimal price vector $\pp'$ for the
new instance will give price $p'_{n+1}=1$ to the $(n+1)$-th item
and price $p'_i=p_i+1$ to each other item $i \in [n]$,
where $\pp$ is an optimal vector for the original instance.


\def\knap{\textsc{Integer Knapsack with repetitions}}

\section{NP-hardness for identical distributions}\label{iid-sec}


In this section we show that \itemp is NP-hard even for 
  identical distributions.   
For this purpose we reduce from the following (still NP-complete) version of Integer Knapsack.

\begin{definition}[\textsc{Integer Knapsack with repetitions}]$ $\\ 
\textsc{Input:} $n+1$ positive integers $a_1<\cdots<a_n$ and $L$.
\\
\textsc{Problem:} Do there exist nonnegative integers $x_1,\ldots,x_n$
  such that $\sum_{i\in [n]} x_i=n$ and
  $\sum_{i\in [n]} x_i a_i =L$?
\end{definition}

The NP-hardness of this version of Integer Knapsack is 
  likely known in the literature, 
  but for completeness we include below a quick proof via a reduction from Subset-Sum.

\begin{lemma}
Integer Knapsack with repetitions is NP-hard.
\end{lemma}
\begin{proof}  
Let $b_1<\cdots<b_n$ and $T$ denote an instance of Subset-Sum,
  where $b_i$ and $T$ are all positive integers.
Without loss of generality, we assume that $T>b_n$.
Let $K=n^2T$.
For each $i\in [n]$, set
  $a_i=K^i+b_i$ {and} $c_i=K^i$. 
Then one can show that $\{K^{n+1},a_i,c_i:i\in [n]\}$, 
  a set of $2n+1$ positive integers, together with 
  $$L=T+K+K^2+\cdots+K^n+(n+1)K^{n+1}$$ form a yes-instance
  of the special Integer Knapsack problem iff a subset of
  $\{b_1,\ldots,b_n\}$ sums to $T$.
\end{proof}

\subsection{Reduction}

Let $a_1<\cdots <a_n$ and $L$ denote an instance of \knap.
Without loss of generality, we assume that $L\le n a_n$; otherwise
  the problem is trivial.
Our goal is to construct a distribution $Q$ over nonnegative integers,
  and reduce the Integer Knapsack problem to the problem
  \itemp with $n$ items, each of which
  has its value drawn from $Q$ independently.
The key idea is similar to the reduction for support size $3$.
$Q$ assigns most of its probability mass to the point $0$,
  so that valuations with many nonzero values contribute very
  little to the expected revenue.
We set the support and probabilities of $Q$ carefully, so that
  the optimization of the expected revenue amounts to a quadratic 
  optimization problem that mimics the Integer Knapsack problem with repetitions.

We start the construction of $Q$ with some parameters.
Let $m=\max(n^5,a_n)$, and let $N=m^{n^2}$ denote a large integer.
For each $i\in [n]$, let $v_i=m^{n+i}$. For each $i\in [n-1]$, let
$$
\g_i=\frac{1}{N}\left(\frac{1}{m^{n+i}}-\frac{1}{m^{n+i+1}}\right)
  =\frac{m-1}{N m^{n+i+1}}.
$$
Let $\g_n=1/(Nm^{2n})$.
For convenience, we also let $\G_i=\sum_{j=i}^n \g_j= {1}/({Nm^{n+i}})$ for each $i\in [n]$.

We record a property that follows directly from our choices of $v_i$ and $\g_i$.
  
\begin{property}\label{val:prob}
For each $i\in [n]$, we have $v_i\G_i=1/N$.
\end{property}

Let now $\qq_1,\ldots,\qq_n$ denote $n$ probability distributions. 
They are closely related to the instance of Integer Knapsack and will be
  specified later in this section.
The support of each $\qq_i$ is a subset of $[2n^3]$ 
  and for each $j\in [2n^3]$, we use $q_i(j)$ to denote the probability
  of $j$ in $\qq_i$.
Finally, let $t_1,\ldots,t_n$ denote a sequence of (not necessarily positive) numbers, also to be 
  specified later, with $|t_i|=O(1/{N^2})$ for all $i\in [n]$.
 
We are ready to define $Q$ using $v_i,\g_i,t_i$ and $\qq_i$.
First, the support of $Q$ is 
$$
\Big\{\hspace{0.05cm}0,v_i,v_i+j:i\in [n]\ \text{and}\ j\in [2n^3]
  \hspace{0.03cm}\Big\}.
$$
Note that all values in the support are bounded 
   by $O(m^{2n})$, and 
  the size of the support is $O(n^4)$. 

Next, $Q$ has probability $(\g_i/m)+t_i$ at $v_i$ for each $i\in [n]$;
  probability $q_i(j)\cdot \g_i(m-1)/m$ at $v_i+j$ for each $i\in [n]$ and $j\in [2n^3]$;
  and probability $1-(\sum_{i=1}^n \g_i+t_i)$ at $0$.
It is easy to verify that $Q$ is a probability distribution since
  the probabilities sum to $1$. 

For convenience, we also let
  $T_i=\sum_{j=i}^n t_j$, and 
  $r_i=\sum_{j=i}^n (\g_j+t_j)=\G_i+T_i$, for each $i\in [n]$.
The latter quantity, $r_i$, is the probability that the value  
  is at least $v_i$.

Even though $t_i$ and $\qq_i$ have not been specified yet, we still can prove 
  the following useful lemma about optimal price vectors, 
  as long as $|t_i|=O({1}/{N^2})$ for each $i\in [n]$:
  
\begin{lemma}\label{pricesupport}
There is an optimal price vector $\pp\in\{v_1,\ldots,v_n\}^n$.
\begin{proof}
By Lemma \ref{searchspace} and Corollary \ref{integerprice} there must be
  an (integral) optimal price vector in $ [0: v_n+2n^3]^n$.

Let $\pp= (p_1,\ldots,p_n)\in [0:v_n+2n^3]^n$ be a price vector with $\pp\notin\{v_1,\ldots,v_n\}^n$. We will prove below that $\Revv(\pp)<\Revv(\bb)$, where $\bb$ is the vector in which all entries are $v_n$.
The lemma then follows.

For convenience, we use $F(s)$ to denote the probability of a 
  random variable drawn from the distribution $Q$ being at least $s$.
For each index $i\in [n]$ such that $p_i>0$, the probability that 
  the buyer picks item $i$ can be bounded from above by $F(p_i)$,
  and can be bounded from below by
$$
F(p_i)\cdot (1-r_1)^{n-1}\ge  F(p_i)\cdot \left(1-
  \frac{1}{m^{n+1}N}-O\left(\frac{n}{N^2}\right)\right)^{n-1}
  \ge F(p_i)-O\left(\frac{n}{m^{2n+2}N^2}\right),
$$
where we used $r_1=\Gamma_1+T_1$, $\Gamma_1=1/(m^{n+1}N)$, $T_1=O(n/N^2)$ 
and $F(p_i) \leq r_1 =O(1/(m^{n+1}N))$ if $p_i >0$.
Applying the upper bound on $\Revv(\pp)$ and the lower bound on 
  $\Revv(\bb)$, we have
$$
\hspace{-1cm}\Revv(\pp)\le \sum_{i:p_i>0} F(p_i)\cdot p_i\ \ \ \text{and}\ \ \ 
\Revv(\bb)\ge nv_n \left(F(v_n)-O\left(\frac{n}{m^{2n+2}N^2}\right)\right)
  \ge nv_nF(v_n)-O\left(\frac{n^2}{m^2 N^2}\right).\hspace{-1cm}
$$

We now examine $p_iF(p_i)$ and $v_nF(v_n)$. We have three cases on $sF(s)$:
\begin{enumerate}
\item[] {\em Case 1}: $s=v_i$ for some $i\in [n]$.
Then we have $$sF(s)=v_i(\Gamma_i+T_i)=\frac{1}{N} \pm O\left(\frac{nm^{2n}}{N^2}\right).$$

\item[] {\em Case 2}: $s=v_i+j$ for some $i\in [n]$ and $j\in [2n^3]$. 
We then have $F(s)\le r_i-(\gamma_i/m)-t_i$ and 
\begin{eqnarray*}
\hspace{-1cm}sF(s)\le (v_i+2n^3)\left(r_i-\frac{\gamma_i}{m}-t_i\right)=\frac{1}{N}\cdot 
  \frac{m^2-m+1}{m^2}+O\left(\frac{n^3}{m^{n+1}N}\right)
  = \frac{1}{N} -\Omega\left(\frac{1}{mN}\right)\hspace{-1cm}
\end{eqnarray*}
when $i<n$, and similarly when $i=n$,
$$sF(s)\le (v_n+2n^3)\cdot \frac{\gamma_n (m-1)}{m}=\frac{m-1}{m}\cdot \frac{1}{N}+O\left(
\frac{n^3}{m^{2n}N}\right)
= \frac{1}{N} -\Omega\left(\frac{1}{mN}\right).$$

\item[] {\em Case 3:} Otherwise, let $i\in [n]$ denote the smallest index such that
  $s<v_i$. Then we have
$$
sF(s)\le (v_i-1)r_i=v_i(\Gamma_i+T_i)-r_i
  =\frac{1}{N}-\Omega\left(\frac{1}{m^{2n}N}\right).
$$


\end{enumerate}

From Case 1, we have $$\Revv(\bb) \geq \frac{n}{N} - O\left(\frac{n^2 m^{2n}}{N^2}\right).$$

Regarding $\Revv(\pp)$,
combining all three cases, we have that 
$$\Revv(\pp) \leq \frac{n}{N} - \Omega\left(\frac{1}{m^{2n}N}\right)$$
because there is at least
one index $i\in[n]$ such that $p_i\notin \{v_1,\ldots,v_n\}$ by the assumption.
As $N \gg n^2 m^{4n}$, we conclude that $\Revv(\pp)<\Revv(\bb)$.
The lemma then follows.
\end{proof}
\end{lemma}

\subsection{Analysis of the Expected Revenue}  
  
Given a price vector $\pp\in \{v_1,\ldots,v_n\}^n$, 
  we let $x_i$ denote the number of items priced at $v_i$. Then $\sum_{i} x_i=n$.  
We will only consider the contribution of two types of valuation vectors 
  to the expected revenue $\Revv(\pp)$: those with exactly one positive 
  entry and those with exactly two positive entries.
The following lemma shows that the 
  total contribution from all other valuation vectors is of third order
  with respect to (roughly) $1/N$.
  
\begin{lemma}\label{triviallemma}
The revenue from valuation vectors with at least three positive entries 
  is $O(n^3/(m^{n+3}N^3))$.
\end{lemma}
\begin{proof}
The probability that a valuation vector has at least three positive entries can 
  be bounded by
$$
O(n^3r_1^3)=O\left(\frac{n^3}{m^{3n+3}N^3}\right). 
$$
Thus, the total contribution is at most
$O(m^{2n})\cdot O(n^3r_1^3)$, and the lemma follows.
\end{proof}  
  
Let $\epsilon= n^3/(m^{n+3}N^3)$ in the rest of the section.
 
Next we examine valuation vectors with exactly one positive entry.  
Their total contribution is 
$$
\sum_{i\in [n]} x_iv_i r_i(1-r_1)^{n-1}.
$$
Since $r_1=O(1/(m^{n+1}N))$ is of first order, approximating the sum up to second order yields
\begin{align}
\sum_{i\in [n]} x_iv_i r_i(1-r_1)^{n-1}&=
  \sum_{i\in [n]}x_iv_ir_i\big(1-(n-1)r_1\pm O(n^2r_1^2)\big) \nonumber\\[0.36ex]
&=\sum_{i\in [n]}x_iv_ir_i - (n-1)\sum_{i\in [n]} x_i
  v_ir_ir_1\pm O(\epsilon).\label{part1}
\end{align}
  
The contribution of valuation vectors with two positive entries is more involved.
First, from those whose two positive entries are over 
  items of the same price, the total contribution to $\Revv(\pp)$ is
\begin{equation}\label{part2}
\sum_{i\in [n]} \frac{x_i(x_i-1)}{2}\cdot v_i\big(r_1^2-(r_1-r_i)^2\big)(1-r_1)^{n-2}.
\end{equation}
For each pair $i<j\in [n]$, we use $p(i,j)\in [0,1]$ to denote the probability 
  of $\alpha-v_i>\beta-v_j$, where $\alpha$ and $\beta$ are drawn independently
  from $Q$ conditioning on $\alpha\ge v_i$ and $\beta\ge v_j$.
Using the $p(i,j)$'s, the contribution from value vectors whose
  two positive entries are over items of different prices is
\begin{equation}\label{part3}
\sum_{i<j\in [n]} x_ix_j\left(v_ir_i(r_1-r_j)+v_jr_j(r_1-r_i) 
+r_ir_j\Big(v_ip(i,j)+v_j\big(1-p(i,j)\big)\Big)\right)(1-r_1)^{n-2}.
\end{equation}

Approximating to the second order, (\ref{part2}) can be simplified to
\begin{equation}\label{part2:2}
\hspace{-0.6cm}\sum_{i\in [n]} \frac{x_i(x_i-1)}{2}\cdot v_i(2r_ir_1-r_i^2)(1\pm O(nr_1))=\sum_{i\in [n]}x_i(x_i-1)v_ir_ir_1
  -\sum_{i\in [n]} \frac{x_i(x_i-1)}{2}\cdot v_ir_i^2 \pm O(\epsilon)
\end{equation}
and (\ref{part3}) can be simplified similarly to
\begin{equation}\label{part3:2}
\sum_{i<j\in [n]} x_ix_j\left(v_ir_i(r_1-r_j)+v_jr_j(r_1-r_i) 
+r_ir_j\Big(v_ip(i,j)+v_j\big(1-p(i,j)\big)\Big)\right)\pm O(\epsilon).
\end{equation}

Next we show that, for each $i\in [n]$, all terms of $v_ir_ir_1$ in (\ref{part1}),
  (\ref{part2:2}) and (\ref{part3:2}) cancel each other.
This is because the overall coefficient of $v_ir_ir_1$ is
$$
-(n-1)x_i+x_i(x_i-1)+\sum_{j:j\ne i} x_ix_j=-(n-1)x_i+nx_i-x_i=0,
$$
where the first equality uses the fact that $\sum_{j \in [n]} x_j = n$.
This allows us to further simplify the sum of
  (\ref{part1}), (\ref{part2:2}) and (\ref{part3:2}), with an error of $O(\epsilon)$, to
\begin{align} \label{total1}
\hspace{-0.3cm}\sum_{i\in [n]}x_iv_ir_i 
-\sum_{i\in [n]}\frac{x_i(x_i-1)}{2}\cdot v_ir_i^2 
-\sum_{i<j\in [n]} x_ix_jr_ir_j(v_i+v_j)+\sum_{i<j\in [n]}
  x_ix_jr_ir_j\big(v_ip(i,j)+v_j(1-p(i,j))\big)
\end{align}
Note that, by Lemma \ref{triviallemma},  
  this is also an approximation of $\Revv(\pp)$, with an error of $O(\epsilon)$.

Let $\epsilon'=n^3m^{n-1}/N^3$.
By plugging in $v_ir_i=v_i(\Gamma_i+T_i)=(1/N)+v_iT_i$ (note
  that $T_i=O(n/N^2)$ is of second order), (\ref{total1}) can be further simplified to
  the following:\vspace{0.04cm}
$$\hspace{-2cm}
\frac{n}{N}+\sum_{i\in [n]} x_iv_iT_i-\sum_{i\in [n]}\frac{x_ir_i(x_i-1)}{2N} -
  \sum_{i<j\in [n]} \frac{x_ix_j(r_i+r_j)}{N}
  +\sum_{i<j\in [n]} \frac{x_ix_j}{N}\cdot\big(r_jp(i,j)+r_i(1-p(i,j))\big)\pm O(\epsilon').\hspace{-2cm}\vspace{0.07cm}
$$
Extracting $x_ix_j(r_i+r_j)/(2N)$ from the last sum above, we get\vspace{0.04cm}
\begin{equation*}
\hspace{-2cm}\frac{n}{N}+\sum_{i\in [n]}x_iv_iT_i-\sum_{i\in [n]}\frac{x_ir_i(x_i-1)}{2N} -
  \sum_{i<j\in [n]} \frac{x_ix_j(r_i+r_j)}{2N}+
  \sum_{i<j\in [n]}\frac{x_ix_j}{N}\cdot\big((1/2)-p(i,j)\big)(r_i-r_j)\pm O(\epsilon').\hspace{-2cm}\vspace{0.07cm}
\end{equation*}
Also note that the second and third sums above
  can be combined into a linear form of the $x_i$'s:
$$
\sum_{i\in [n]} x_ir_i(x_i-1)+\sum_{i<j\in [n]}x_ix_j(r_i+r_j)
  =-\sum_{i\in [n]}x_ir_i+\left(\sum_{i\in [n]} x_i\right)\left(\sum_{i\in [n]}x_ir_i\right)
  =(n-1)\sum_{i\in [n]}x_ir_i.
$$
As a result, we get the following approximation of the expected revenue
  $\Revv(\pp)$, with an error of $O(\epsilon')$:
\begin{equation}\label{haha2}
\frac{n}{N}+\sum_{i\in [n]}x_iv_iT_i-\frac{n-1}{2N}\sum_{i\in [n]}x_ir_i
+\sum_{i<j\in [n]} \frac{x_ix_j}{N}\cdot \big((1/2)-p(i,j)\big)(\Gamma_i-\Gamma_j).
\end{equation}
Note that in (\ref{haha2}), we also replaced $r_i-r_j$ at the end with $\Gamma_i-\Gamma_j$
  since the error introduced is $O(n^3/N^3)$.

\subsection{Reverse Engineering of $t_i$ and $p(i,j)$}

Our ultimate goal is to set $t_i$'s and $\qq_i$'s 
  carefully so that (\ref{haha2}) by the end has the following form:
\begin{equation}\label{final1}
\frac{n}{N}+\frac{L^2}{N^2m^{3n}}-\frac{1}{N^2m^{3n}}\cdot \left(\sum_{i\in [n]} x_ia_i-L\right)^2.
\end{equation}
Recall that $L$ is the target integer in the Knapsack instance.
If this is the case, then we obtain a polynomial-time reduction from 
  the special Knapsack problem to \textsc{Item-Pricing},
  since the difference between (\ref{final1}) and $\Revv(\pp)$ is 
  at most $O(\epsilon')$ and thus (\ref{final1}) is at least
$$
\frac{n}{N}+\frac{L^2}{N^2m^{3n}}-\frac{1}{2N^2m^{3n}}
$$
if and only if $a_1,\ldots,a_n$ and $L$ is a yes-instance of the 
  special Knapsack problem.
  
To compare (\ref{final1}) and (\ref{haha2}),
  we use $\sum_{i\in [n]} x_i=n$ in (\ref{final1}) and it becomes
\begin{align}
&\frac{n}{N}-\frac{1}{N^2m^{3n}}\cdot \left(\sum_{i\in [n]} x_i^2a_i^2+
2\sum_{i<j\in [n]}x_ix_ja_ia_j-2\sum_{i\in [n]}a_iLx_i\right) \nonumber\\[0.36ex]
&=\frac{n}{N}-\frac{1}{N^2m^{3n}}\cdot \left(\sum_{i\in [n]} a_i^2 x_i
  \left(n-\sum_{j:j\ne i} x_j\right)+
2\sum_{i<j\in [n]}x_ix_ja_ia_j -2\sum_{i\in [n]}a_iLx_i\right) \nonumber\\[0.36ex]
&=\frac{n}{N}-\frac{1}{N^2m^{3n}}\cdot \left(\sum_{i\in [n]} (na_i^2-2a_iL) x_i
  -\sum_{i<j\in [n]}x_ix_j(a_i-a_j)^2 \right)\label{haha1}
\end{align} 
By comparing (\ref{haha1}) with (\ref{haha2}), our goal is achieved if
  the following two conditions hold: First,
\begin{equation}\label{temp}
T_i=\frac{1}{v_i}\cdot\left(\frac{(n-1)r_i}{2N}-\frac{1}{N^2m^{3n}}\cdot \left(na_i^2-2a_iL\right)\right),
\end{equation}
for all $i\in [n]$ (note that the absolute value of the
  right side of (\ref{temp}) is $O(n/(m^{2n+2}N^2))$; Second,
\begin{equation}\label{temp22}
\frac{((1/2)-p(i,j))(\Gamma_i-\Gamma_j)}{N}=\frac{(a_i-a_j)^2}{N^2m^{3n}},\ \ \ \ 
\text{for all pairs $i<j\in [n]$.}
\end{equation} 

For the first condition, we note that the equations (\ref{temp})
for all $i \in [n]$ actually form a triangular
system of $n$ equations in the $n$ variables $t_1,\ldots,t_n$,
and thus there exists a unique
  sequence $t_1,\ldots,t_n$ such that (\ref{temp}) holds for all $i\in [n]$.
Moreover, as the absolute value of the 
  right side of (\ref{temp}) is $O(n/(m^{2n+2}N^2))$, the $t_i$'s
  are $O(1/N^2)$ as we promised earlier.
To see this, we let $s$ denote the maximum of the absolute value of the
  right side of (\ref{temp}), over all $i\in [n]$.
Then one can show by induction on $i$ that $|t_i|\le 2^{n-i}s$ for all $i$ from $n$ to $1$.
The claim now follows using $2^n\ll m^n$.

The second condition is more difficult to satisfy.
From (\ref{temp22}), we know that the condition is met if
\begin{equation}\label{targettarget}
\frac{1}{2}-p(i,j)=\frac{(a_i-a_j)^2}{Nm^{3n}(\Gamma_i-\Gamma_j)},\ \ \ \ 
\text{for all $i<j\in [n]$.}
\end{equation}
We will define below the $n$ distributions $\qq_i$, $i\in[n]$,
so that their induced values for the probabilities $p(i,j)$ satisfy
  (\ref{targettarget}).
An important property that we will need for the construction of the $\qq_i$'s
is that all the desired probabilities $p(i,j)$ are very close to 1/2.
Specifically, using $\Gamma_i-\Gamma_j\ge \gamma_i\ge \gamma_n=1/(m^{2n}N)$, we have
\begin{equation}\label{imp1}
0<\frac{1}{2}-p(i,j)\le \frac{(a_i-a_j)^2\cdot Nm^{2n}}{Nm^{3n}}=o\left(\frac{1}{m}\right),
\end{equation}
since $m=\max\hspace{0.05cm}(n^5,a_n)$ and $a_n=\max_{i\in [n]} a_i$.

\subsection{Connecting $p(i,j)$ with $\qq_i$ and $\qq_j$}

Fixing a pair $i<j\in [n]$,
  we examine $p(i,j)$ closer. Recall that $p(i,j)$ is the probability of
  $\alpha-v_i>\beta-v_j$ when $\alpha$ and $\beta$ are drawn independently
  from $Q$, conditioning on $\alpha\ge v_i$ and $\beta\ge v_j$.

For convenience, we use \emph{block} $k$ to denote the subset $\{v_k,v_k+1,\ldots,v_k+2n^3\}$
  of the support of $Q$.
Note that due to the exponential structure of the support of $Q$
  (and the assumption of $i<j$),
if $\alpha$ is in block $k \geq i$ and $\beta$ is in block $\ell > j$ 
with $\ell >k$ then $\beta-v_j > \alpha-v_i$.
Therefore, for $\alpha-v_i>\beta-v_j$ to happen,
  we only need to consider the following three cases:
\begin{flushleft}\begin{description}
\item[\ \ \ \ Case 1:] $\alpha$ is from block $k$ and $\beta$ is from block $\ell$,
  where $k,\ell\in [n]$ satisfy $k\ge\ell > j$.
Then the total contribution of this case to probability $p(i,j)$ is:
$$
\frac{1}{r_ir_j}\cdot \sum_{k\ge\ell >  j} (\gamma_k+t_k)(\gamma_\ell+t_\ell).
$$
\item[\ \ \ \ Case 2:] $\alpha$ is from block $k$ and $\beta$ is from block $j$,
  where $k>i$. Then the total contribution is
$$
\frac{1}{r_ir_j}\cdot \sum_{k>i} (\gamma_k+t_k)(\gamma_j+t_j).
$$
\item[\ \ \ \ Case 3:] Finally, $\alpha$ is from block $i$ and $\beta$ is from block $j$, with
  $\alpha-v_i>\beta-v_j$.
Let $q(i,j)$ denote the probability of $\alpha>\beta$,
  when $\alpha$ is drawn from $\qq_i$ and $\beta$ is drawn from
  $\qq_j$ independently.
Using $q(i,j)$, the total contribution of this case to $p(i,j)$ is
$$
\frac{1}{r_ir_j}\cdot\left(
\left(\frac{\gamma_j}{m}+t_j\right)\cdot \frac{(m-1)\gamma_i}{m}+q(i,j)\cdot
  \frac{(m-1)\gamma_i}{m}\cdot \frac{(m-1)\gamma_j}{m}\right).
$$
\end{description}\end{flushleft}

The probability $p(i,j)$ is equal to the sum of the above three quantities for the three cases.
Hence,  $q(i,j)$ is uniquely determined by the $p(i,j)$
  we aim for, i.e., the unique $p(i,j)$ that satisfies (\ref{targettarget}), because all other
  parameters have been well defined by now, including $t_1,\ldots,t_n$.
  
We show below that, if $\left|\hspace{0.036cm}p(i,j)-1/2\hspace{0.036cm}\right|=o(1/m)$, then the 
  $q(i,j)$ it uniquely determines must 
  satisfy $\left|\hspace{0.036cm}q(i,j)-1/2\hspace{0.036cm}\right|=O(1/m)$.
  
To see this, note first that since $i<j\le n$ and $r_i=\Gamma_i+T_i=1/(Nm^{n+i})\pm O(n/N^2)$,
  we have that 
$$\gamma_i=\frac{m-1}{m^{n+i+1}N}=\frac{m-1}{m}\cdot r_i \pm O\left(\frac{n}{N^2}\right).$$
Thus, $\sum_{k>i}(\gamma_k+t_k)=r_i-\gamma_i-t_i=r_i/m\pm O(n/N^2)$.

Using this fact in the above expressions for the three cases, 
it is easy to show that, other than 
\begin{equation}\label{finalterm}
\frac{1}{r_ir_j}\cdot q(i,j)\cdot \frac{(m-1)\gamma_i}{m}\cdot \frac{(m-1)\gamma_j}{m},
\end{equation}
the contribution of other terms is bounded from above by $O(1/m)$
  (note that $k\ge \ell> j$ implies $k>i$).
Since $\left|\hspace{0.036cm}p(i,j)-1/2\hspace{0.036cm}\right|=o(1/m)$, 
it follows that 
  the term in (\ref{finalterm}) is between
  $1/2-O(1/m)$ and $1/2+O(1/m)$.
Note that $\gamma_i = (m-1)r_i/m \pm O(n/N^2)$ 
(since $i <n$), and $\gamma_j $
is either $(m-1)r_j/m \pm O(n/N^2)$ if $j<n$ or $r_j \pm O(n/N^2)$ if $j=n$.
Therefore, the coefficient of $q(i,j)$ in (\ref{finalterm}) is $1-O(1/m)$.
Since the expression in (\ref{finalterm}) is $1/2 \pm O(1/m)$,
it follows that 
$\left|\hspace{0.03cm}q(i,j)-1/2\hspace{0.03cm}\right|=O(1/m)$.

\subsection{Reverse Engineering of $\qq_i$}

Given $q(i,j)$ for each pair $i<j\in [n]$,
  our final technical step of the reduction is to construct a sequence of probability distributions
  $\qq_1,\ldots,\qq_n$ over $[2n^3]$ such that, for each pair $i<j\in [n]$,
  the probability of $\alpha>\beta$, where $\alpha$ is drawn from $\qq_i$
  and $\beta$ is drawn from $\qq_j$ independently, is exactly 
  $q(i,j)$.
  
In general, such a sequence of distributions may not 
  exist, e.g., consider $n=3$, $q(1,2)=1$, $q(2,3)=1$ and $q(1,3)=0$. 
  But here we are guaranteed that the
  $q(i,j)$'s are close to $1/2$: $|\hspace{0.03cm}q(i,j)-1/2\hspace{0.03cm}|=O(1/m)$.
  We shall show that in this case the desired distributions exist, and we can construct them.

To construct $\qq_1,\ldots,\qq_n$, we define ${n\choose 2}$
  subsets of $[2n^3]$, called \emph{sections}. 
Each section consists of $2n+3$ consecutive integers.
The first section is $\{1,\ldots,2n+3\}$, the second
  section is $\{2n+4,\ldots,4n+6\}$, and~so on and so forth.
(Note that $2n^3$ is clearly large enough for ${n\choose 2}$ sections.)
Each section is labeled, arbitrarily, by a distinct pair $(i,j)$ with $i<j\in [n]$.
We let $t_{i,j,k}$ denote the $k$th smallest integer
  in section (labeled) $(i,j)$, where $k\in [2n+3]$.
Now we define $\qq_\ell$, $\ell\in [n]$. For each section $(i,j)$,
  $i<j\in [n]$, we have:
\begin{description}
\item[\ \ \ \ Case 1:] If $\ell\ne i$ and $\ell\ne j$, then we set
$$q_\ell\big(t_{i,j,\ell}\big)=q_\ell\big(t_{i,j,2n+4-\ell}\big)=\frac{1}{2{n\choose 2}}$$
and $q_\ell\big(t_{i,j,k}\big)=0$ for all other $k\in [2n+3]$.

\item[\ \ \ \ Case 2:] If $\ell=j$, then we set $q_\ell\big(t_{i,j,n+2}\big)=1/ {n\choose 2}$ and
  $q_\ell\big(t_{i,j,k}\big)=0$ for all other $k\in [2n+3]$.
  
\item[\ \ \ \ Case 3:] If $\ell=i$, then we set 
$$
q_{\ell}\big(t_{i,j,n+1}\big)=\frac{1}{2{n\choose 2}}-{n\choose 2}
  \big(q(i,j)-1/2\big)\ \ \ \text{and}\ \ \  
q_\ell\big(t_{i,j,n+3}\big)=\frac{1}{2{n\choose 2}}+{n\choose 2}
  \big(q(i,j)-1/2\big),
$$
and $q_\ell\big(t_{i,j,k}\big)=0$ for all other $k\in [2n+3]$.
\end{description}
This finishes the construction of $\qq_1,\ldots,\qq_n$.
Using $|\hspace{0.03cm}q(i,j)-1/2\hspace{0.03cm}|\hspace{-0.02cm}=
  \hspace{-0.02cm}O(1/m)$ and $m\ge n^5$, we know that $\qq_1,\ldots,\qq_n$ are 
  probability distributions: all entries are nonnegative and sum to $1$.
  
It is also not hard to verify that the distributions satisfy the desired property,
i.e., for each pair $i<j\in [n]$
  the probability of $\alpha>\beta$, where $\alpha$ is drawn from $\qq_i$
  and $\beta$ is drawn from $\qq_j$ independently, is exactly 
  $q(i,j)$. 
  First observe that every section of each distribution $\qq_i$ has the
  same probability $1/{n\choose 2}$. If $\alpha$ and $\beta$
  belong to different sections then the order between $\alpha$ and $\beta$
  is determined by the order of the sections, and both orders have
  obviously the same probability.
  
  So suppose that $\alpha, \beta$ belong to the same section labeled $(g,h)$,
    where $g,h\in [2n+3]$.
  If $g \neq i$ or $h \neq j$, then it is easy to check that both orders
  between $\alpha$ and $\beta$ have the same probability.
  To see this, suppose first that $i \not\in \{g,h\}$. 
  Then $\alpha = t_{g,h,i}$ or $t_{g,h,2n+4-i}$ with equal probability. 
  If $\alpha = t_{g,h,i}$, then
  $\alpha < \beta$ because $\beta$ is 
  either $t_{g,h,j}$ or $t_{g,h,2n+4-j}$ (if $j \not\in \{g,h\}$),
  or $\beta = t_{g,h,n+2}$ (if $h=j$) 
  or $\beta = t_{g,h,n+1}$ or $t_{g,h,n+3}$ (if $g=j$);
  similarly, if $\alpha =t_{g,h,2n+4-i}$ then $\alpha > \beta$.
  Therefore, if $i \not\in \{g,h\}$, then there is equal probability
  that $\alpha < \beta$ and $\alpha > \beta$.
  Similarly, the same is true if $j \not\in \{g,h\}$.
  
  Suppose that $i \in \{g,h\}$ and $j \in \{g,h\}$. Since $i<j$ and $g <h$,
  we must have $i=g$ and $j=h$.
  In this case, $\beta = t_{i,j,n+2}$, and
  $\alpha = t_{i,j,n+1}$ or $\alpha = t_{i,j,n+3}$,
  hence $\alpha > \beta$ iff $\alpha = t_{i,j,n+3}$.
  
  The probability that $\alpha>\beta$ and $\alpha, \beta$ are not both in section
  $(i,j)$ is $$\frac{1}{2} \cdot \left( 1 - \frac{1}{ {n \choose 2}^2 } \right).$$
  The probability that $\alpha>\beta$ and $\alpha, \beta$ are both in section
  $(i,j)$ is $$\frac{1}{ {n \choose 2}} \cdot \left( \frac{1}{ 2{n \choose 2}} + {n \choose 2} 
  \left(q(i,j) -\frac{1}{2}\right) \right)= \frac{1}{ 2{n \choose 2}^2 } + \left(q(i,j) - \frac{1}{2}\right).$$
  Thus, the total probability that $\alpha>\beta$ is exactly $q(i,j)$ as desired.

This concludes the construction and the proof of the theorem. 


\section{Conclusions}

In this paper, we studied the complexity of the Bayesian Unit-Demand Item-Pricing problem with independent distributions. 
We showed that the decision problem is NP-complete even when the distributions are 
  of support size $3$ or when they are identical. We also presented a 
polynomial-time algorithm for distributions of support size $2$.

Several interesting open questions remain. Is there a PTAS for general distributions?
Note that our NP-hardness results do not preclude the existence of an FPTAS. 
Actually, by adapting techniques from \cite{CD11} we can give an FPTAS 
  for the case 
when the supports of the distributions are integers in a bounded interval. Moreover, we conjecture that the IID case can be solved in polynomial time
when the size of the support is constant.

A related question concerns the complexity of the randomized case (i.e., lottery pricing). We conjecture that this problem is intractable, but new ideas
are needed to prove this.


\bibliographystyle{alpha}
\bibliography{references}
\end{document}